\documentclass[aoas,preprint]{imsart}

\RequirePackage{amsthm,amsmath,amsfonts,amssymb}
\RequirePackage{natbib}
\RequirePackage{graphicx}
\RequirePackage{booktabs}
\RequirePackage{xcolor}
\RequirePackage[colorlinks,citecolor=blue,urlcolor=blue,linkcolor=blue]{hyperref}

\startlocaldefs
\theoremstyle{plain}
\newtheorem{proposition}{Proposition}

\theoremstyle{definition}

\DeclareMathOperator{\logit}{logit}
\newcommand{\R}{\mathcal{R}}
\newcommand{\N}{\mathcal{N}}
\newcommand{\ind}[1]{\mathbf{1}\{#1\}}
\newcommand{\D}{\mathcal{D}}
\newcounter{experiment}
\newcommand{\expt}[2]{\par\medskip\refstepcounter{experiment}\label{#2}%
\noindent\textbf{Analysis~\theexperiment\ (#1).}\hspace{0.4em}\ignorespaces}
\makeatletter
\g@addto@macro\common@pub@types{}
\makeatother
\endlocaldefs

\begin{document}

\begin{frontmatter}
% Title settled by the author 3 Oct 2026. The 16 Sep AOAS reframe had used a
% question form, which repeated itself: the subtitle restated the question in
% jargon. This keeps the applied purpose in the tail, so the application reads as
% the point rather than as an appended use case, which was the shape diagnosed
% behind the Bayesian Analysis and JCGS desk rejections. "Subgraphs" is one word
% for what "graph-structured subsets" says in three.
\title{Calibrated Bayesian Selection of Schema Subgraphs for
Database Question Answering}
\runtitle{Calibrated Bayesian Schema Selection}

\begin{aug}
\author[A]{\fnms{Robert}~\snm{Richardson}\ead[label=e1]{richardson@stat.byu.edu}}
\address[A]{Department of Statistics, Brigham Young University\printead[presep={,\ }]{e1}}
\end{aug}

\begin{abstract}
An LLM answering a question about a relational database must know which tables hold the information.
Showing too few makes the question unanswerable, while showing too many wastes context. Most schema-linking
methods score each table alone, though the required tables are usually connected by foreign keys and can
include connectors matching the question only weakly. We develop a hierarchical Bayesian model that combines
question evidence with the foreign-key graph and returns empirically near-calibrated probabilities over complete sets of tables, not a single
list, so later stages can tell when a selection is uncertain. Schema graphs are small
enough to evaluate every subset exactly. With a fixed context budget the graph improves the chance of
retaining every required table, and the benefit grows with schema size, turning positive at about seven
tables. When the model also chooses how many tables to keep, a learned prior on subset size accounts
for the whole improvement and the graph adds nothing. Fit the size prior first, and use the graph when the
budget is fixed. The model's predictive sets cover the true set about as often as promised, where
independent scoring is overconfident.
\end{abstract}


\end{frontmatter}

\section{Introduction}

A large language model answering a question against a database is only as good as the schema it is shown.
Pass the whole schema and you waste context and degrade the answer, since models use long contexts unevenly
and are distracted by irrelevant material \citep{liu2024lost}. Pass too little and the query cannot be
answered at all. Schema linking, the task of selecting the relevant part of the schema, is therefore often the
difference between a correct and an incorrect query, and a wrong selection propagates to later steps that have
no way to judge how much to trust it \citep{kuhn2023}. Because the tables a query reads are joined by foreign
keys, the relevant part of the schema is typically a connected subgraph rather than an arbitrary set.

How much this matters depends on scale, and the widely used benchmarks understate it. Spider
\citep{yu2018spider} and BIRD \citep{li2023bird} are collections of questions paired with
reference queries, each posed against its own database. Those databases hold four and eight tables at the
median, sizes at which a reasonable context budget can simply include every table. BEAVER \citep{chen2024beaver} was assembled from
enterprise systems for that reason, and its largest schema holds ninety-seven tables. Schema linking presumes the schema
must be cut down. At four tables it need not be, so a method tested only on the small public benchmarks
is measured where the task it solves is least necessary. We analyse $639$ multi-table questions over $8$ BIRD databases and $459$ over $20$ Spider databases, together with the $97$-table enterprise warehouse of \S\ref{sec:beaver}.

For large schemas the literature has turned to agents. AutoLink \citep{autolink} prunes autonomously
and reports severe degradation of competing methods beyond three thousand columns, EviLink \citep{evilink}
spends an uncertainty budget deciding where to gather further evidence, and graph-structured retrieval
assembles a working context before generation rather than passing a corpus wholesale \citep{edge2024graphrag}.
What an agent changes is how often the decision is revisited, not that something must still decide which
schema information to bring into view under uncertainty about what the query needs. A better static selector
therefore improves the agentic workflow as well.

Most schema linkers score each candidate independently and threshold or rank \citep{evilink,subgraphrag}. That
works when every relevant table is on its own clearly similar to the question, and it misses a common pattern.
A join table, or one carrying a filter the question only implies, can be relevant while matching the question
weakly. We call such a table \emph{weakly-signaled}. The important special case is a structural
connector, relevant because it joins other clearly relevant tables. Figure~\ref{fig:concept} shows what goes
wrong: scoring tables alone keeps the two endpoints the question names and drops the junction between them.

\begin{figure}[t]\centering
\includegraphics[width=0.95\linewidth]{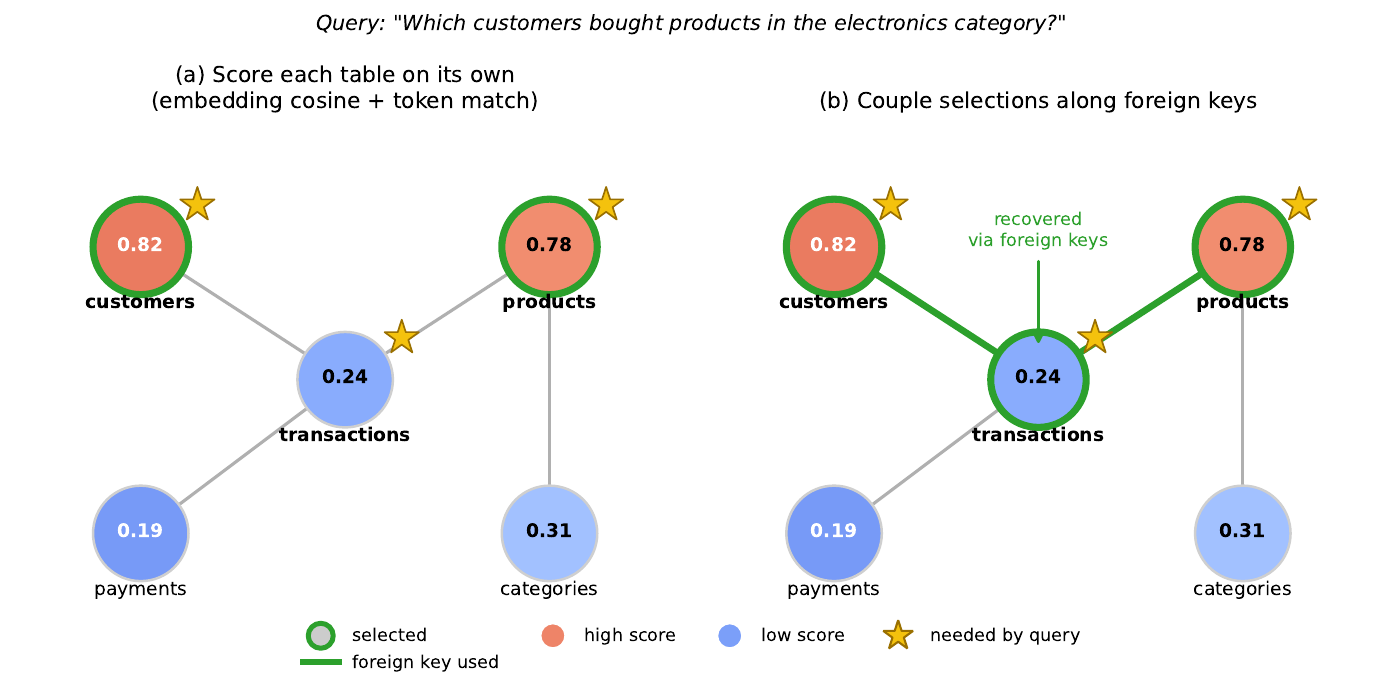}
\caption{Schema linking as subgraph selection. Nodes are tables and an edge joins two tables related by a
foreign key. Each node's shade and number is the per-table score $\sigma(a_t)$ from query--table embedding
cosine and token overlap, and $\star$ marks a table the query needs. (a) Scoring tables on their own selects
the two clearly relevant endpoints but misses the weakly-signaled junction table \texttt{transactions}
(score $0.24$). (b) Coupling along the foreign keys lifts it over threshold, because both neighbors are
selected, and returns the connected subgraph the query requires.}\label{fig:concept}
\end{figure}

Recent uncertainty-aware methods keep independent scoring in the object they infer. EviLink's credibility is an
independent Beta--Binomial per item, so what it carries forward is itemwise uncertainty. SubgraphRAG
\citep{subgraphrag} scores triples with independent Bernoullis and uses the graph only as a feature, and
SchemaGraphSQL \citep{schemagraphsql} joins endpoint tables by deterministic path-finding. Relation-aware
encoders such as RAT-SQL \citep{wang2020ratsql} fold structure into the representation but still decode a
single parse. Our target is a joint calibrated distribution over schema subsets, meaning one whose stated
probabilities match the frequencies they refer to, so that among tables called $70\%$ likely, about $70\%$ truly
are. That is what lets a later step defer to the user when evidence is thin, and it is what a deterministic
selector cannot supply, though if a single high-precision set is all that is wanted a shortest-path heuristic
remains a strong alternative (Table~\ref{tab:selectors}).

We therefore leave the representation alone and put the Bayesian model on the selection. Embedding and token
scores enter as fixed unary evidence, a per-table score depending only on that table's own match to the
question. A pairwise term on the foreign-key edges then rewards selecting adjacent tables together. The result is an autologistic Markov random field \citep{besag1974}.

The model sits in an established line of work on structured Bayesian selection. Binary indicators coupled by an Ising or Potts
term are standard in Bayesian variable selection \citep{george1993,lizhang2010,stingo2011}, where the coupling
is usually fixed to avoid a phase transition in which the model size explodes. Closest to our approach,
\citet{smith2007spatial} estimate a spatial selection coupling but keep it global, \citet{moores2020scalable}
give principled inference for a Potts inverse temperature, again global, and \citet{peterson2015bayesian} let a
cross-graph coupling vary by group-pair while driving edge selection. We adapt that group-varying idea from
edge to node selection and replace the per-pair coupling with a partial-pooling hierarchy on a scalar. In another similar approach, \citet{lin2026spatial} perform spatially adaptive binary selection with a conditional autoregressive
prior. We differ in the object reported and in what varies. They select spatial locations under a fixed
neighborhood structure, whereas we report a calibrated joint distribution over whole subsets, let the coupling
vary by database, evaluate the conditional likelihood exactly, and score selection
against retained context under a deployment-matched protocol. The
overconfidence we document also has precedent, since structured selection posteriors can be miscalibrated, as
in genetic fine-mapping \citep{cui2024finemapping,hutchinson2020}. We assess it with reliability diagrams,
expected calibration error \citep{guo2017calibration}, proper scores \citep{brier1950,gneiting2007}, and
post-hoc scaling \citep{platt1999} as a baseline. Because schema graphs are small, the likelihood and every
inclusion probability are computed exactly. The contributions are as follows.

\begin{enumerate}\itemsep2pt
\item \textbf{A hierarchical joint model with exact inference.} To our knowledge the first coherent joint
Bayesian posterior over a connectivity-favoring schema subset, with database-level coupling and, because schema
graphs are small, an exactly evaluated conditional likelihood.
\item \textbf{A calibration mechanism.} Coupling improves recovery but makes marginal inclusion probabilities
overconfident. We identify this as coupling-induced mean shift, corrected by a compensating decrease in the
intercept (Proposition~\ref{prop:intercept}), which is why the fitted intercept is negative and why temperature
scaling, unable to represent an additive correction, fails here.
\item \textbf{Where the graph helps, by how much, and where it does not.} We measure the size of the
advantage of the approach under various scenarios. It shows clear improvements when the number of retrieved tables is fixed or when the schema is large, such as seven tables or more. It adds little when the selector chooses the number of relevant tables, where a prior on set size
accounts for the whole gain and an independent selector given that prior matches the coupled one. We report
both for a database never seen and for one whose queries have already accumulated, since both occur in
practice.
\end{enumerate}
An inferred-edge experiment shows the method can operate without declared foreign keys when a
cross-database edge model is available, though it does not improve on declared keys, which is a
robustness check rather than a fourth contribution.

The method is not proposed as a way to raise end-to-end execution accuracy. Its value is the posterior itself, which is what a later stage needs in order to defer, ask, or widen its selection when the evidence is thin.
\S\ref{sec:theory} gives the recovery and calibration results, \S\ref{sec:experiments} the evidence, and
\S\ref{sec:disc} states which parts of the Bayesian treatment earn their place and which do not.

\section{Model and inference}\label{sec:model}

Let $G=(V,E)$ be the schema graph, where $V$ indexes the tables and $E$ contains one edge for each foreign-key
relationship between two tables. The graph is read directly from the database catalogue, whose tables are the
nodes and whose declared foreign-key constraints are the edges, so it is given, not estimated. A selection
is a binary vector $x\in\{0,1\}^{|V|}$, with $x_t=1$ meaning table $t$ is kept. We turn each table's resemblance to the question into a per-table score (Section \ref{sec:evidence}),
assemble those scores into a conditional model over selections that also rewards connectivity (Section
\ref{sec:coupling}), and place priors and database-level random effects on its parameters so that we can report
a posterior predictive distribution (Section \ref{sec:hier}).

\subsection{Per-table evidence}\label{sec:evidence}
Evidence that a table is relevant is how well its text matches the question. We use six fixed features per
table, embedding and lexical. These are all deterministic given the tables, not estimated.  The embedding features are cosine
similarities between a pretrained embedding of the question and embeddings of each table's name and column
names, which capture meaning where words differ, so that a question about earnings matches a column named
\texttt{salary}. The lexical features are a BM25 score, token overlap, and a value-match indicator
for whether a question word appears among the table's stored values, sampled from a $500$-row scan. Because
embeddings blur exact matches and lexical scores miss paraphrase, the two are complementary.

These combine into one score per table, $a_t=\theta^\top\phi_t$ with $\theta$ fitted by logistic regression,
so $\sigma(a_t)$ estimates the probability that table $t$ is relevant from that table's text alone. We
call this the \emph{unary} evidence, and ranking by $\sigma(a_t)$ above a threshold is the \emph{independent
selector}. Feature construction and prior scales are in \hyperref[s:app:checks]{Supplement~\ref*{s:app:checks}}. The limitation is structural rather
than one of feature quality. For the junction table of Figure~\ref{fig:concept}, $\sigma(a_t)$ falls below one
half and no unary score can separate it from unrelated tables, because a unary score has no way to use the
fact that it sits between two tables the question clearly needs.

\subsection{Coupling on the graph}\label{sec:coupling}

Relevant tables tend to be joined to one another, and we encode this dependence with a pairwise Markov random
field on the graph. Given parameters $\psi=(\alpha,\gamma,\beta)$, the model over the selection is
\begin{equation}\label{eq:mrf}
p(x\mid\phi,G,\psi)\;=\;\frac{1}{Z(\phi,G,\psi)}\,
\exp\!\Big(\alpha\textstyle\sum_{t} x_t+\gamma\sum_{t} a_t x_t+\beta\sum_{(s,t)\in E} x_s x_t\Big),
\qquad \beta\ge 0 .
\end{equation}
Equation~\eqref{eq:mrf} is a conditional (autologistic) model of the selection given the
evidence $\phi$ and the graph $G$, not a prior on $x$ that is subsequently updated by data, so the graph coupling
is part of the model specification rather than something the data revise. The same expression does serve as
a likelihood in \S\ref{sec:hier}, but for the parameters $\Theta$: it is $\Theta$ that carries a prior and a
posterior there, never $x$. The first two terms carry the per-table evidence, where
$\gamma$ scales the unary scores and $\alpha$ is an intercept setting how many tables are selected on average.
With $\beta=0$ the tables are conditionally independent, and the marginal inclusion probability is
$\sigma(\alpha+\gamma a_t)$. Taking in addition $\alpha=0$ and $\gamma=1$ recovers the independent selector
$\sigma(a_t)$. In the coupling ablations we use the same model with $\beta=0$ as the independent baseline so
that the comparison is exact. The third term
adds $\beta$ to the score of every pair of foreign-key-adjacent tables selected together, so the model favors
selections that are connected on the schema graph.

The log-odds that table $t$ is relevant, given the current
status of the other tables, is $\gamma a_t+\alpha+\beta\sum_{s\in\N(t)}x_s$, the table's own evidence plus
$\beta$ for each already-selected foreign-key neighbor $\N(t)$. A weakly-signaled table, whose own evidence
puts it below one half, can be lifted above the threshold once enough of its neighbors are in. This is not only
a schematic device. Section~\ref{sec:experiments} demonstrates it on a real database, where a
junction table that the cosine ranks below threshold is pulled above it once a foreign-key neighbor is
selected, and the sequential form of the model stops at exactly the gold set of tables (the trace plot in \hyperref[s:app:extra]{Supplement~\ref*{s:app:extra}}).

\paragraph*{Connectivity}
The pairwise term favors connected selections but does not enforce connectedness, so
\eqref{eq:mrf} places some mass on disconnected sets. We do not impose a hard connectivity constraint because the gold table set is already connected in the foreign-key graph in 95\% of BIRD
queries and 93\% of Spider queries, and the fitted posterior's selection is connected 92 to 96\% of the time.
A hard-connected variant, which restricts the enumeration to configurations whose selected set is connected, is
straightforward and at the fitted coupling improves recovery, as \hyperref[s:app:extra]{Supplement~\ref*{s:app:extra}} reports. We keep the
unconstrained model as primary because it can represent the disconnected gold sets that occur in the 5 to 7\% of
queries where the relevant tables are not connected, which the hard constraint cannot.

\subsection{A hierarchical Bayesian model}\label{sec:hier}

\emph{Likelihood.} The data are the gold table sets $x^\star_q$ observed on a set of training queries, each
query $q$ posed against a database $d(q)$ with schema graph $G_d$ and features $\phi_q$. Given parameters
$\Theta$, the selection for query $q$ follows the autologistic model \eqref{eq:mrf}, written per query and with
$\gamma=1$ as
\begin{equation}\label{eq:lik}
p(x_q\mid\phi_q,G_d,\Theta)=\frac{1}{Z(\phi_q,G_d,\Theta)}\exp\!\Big(\sum_t(\alpha_d+\theta^\top\phi_{qt})x_{qt}
+\beta_d\!\!\sum_{(s,t)\in E_d}\!\!x_{qs}x_{qt}\Big).
\end{equation}
Fixing $\gamma=1$ removes the scale non-identifiability between $\theta$ and $\gamma$, since multiplying $\theta$
by $c$ and dividing $\gamma$ by $c$ leaves the model unchanged, so the natural parameters are
$(\theta,\alpha_d,\beta_d)$. Queries are conditionally independent given $\Theta$, and the likelihood of the
observed selections is $\prod_q p(x^\star_q\mid\phi_q,G_{d(q)},\Theta)$.

\emph{Prior.} The unary weights $\theta$ are shared across databases, while the intercept and coupling vary by
database through non-centered random effects,
\begin{align}\label{eq:prior}
&\theta\sim\N(0,4I), \qquad
\alpha_d=\alpha_0+\sigma_\alpha z^\alpha_d, \qquad
\beta_d=\operatorname{softplus}(b_0+\sigma_\beta z^\beta_d),\notag\\
&(z^\alpha_d,z^\beta_d)\sim\N(0,I_2), \qquad
\alpha_0,b_0\sim\N(0,9), \qquad
\sigma_\alpha,\sigma_\beta\sim\text{Half-Normal}(1),
\end{align}
which lets densely and sparsely connected schemas take different couplings while borrowing strength through the
shared hyperparameters $(\alpha_0,b_0,\sigma_\alpha,\sigma_\beta)$.

\emph{Posterior.} Writing $\Theta=(\theta,\alpha_0,b_0,\sigma_\alpha,\sigma_\beta,\{z_d\})$, Bayes' theorem gives
\begin{equation}\label{eq:post}
p(\Theta\mid\D)\ \propto\ \Big[\textstyle\prod_q p(x^\star_q\mid\phi_q,G_{d(q)},\Theta)\Big]\,p(\Theta).
\end{equation}
Because the schema graphs are small, at most fourteen tables here, the normalizer $Z$ in \eqref{eq:lik} is
evaluated exactly by enumerating all $2^{|V|}$ configurations, so \eqref{eq:post} is a posterior under an exactly
computed conditional likelihood rather than a pseudo-likelihood.

\emph{Posterior predictive.} For a new query the object of interest is the posterior predictive distribution over
its selection,
\begin{equation}\label{eq:predictive}
p(x_{\mathrm{new}}\mid\phi_{\mathrm{new}},G_{\mathrm{new}},\D)=\int p(x_{\mathrm{new}}\mid\phi_{\mathrm{new}},
G_{\mathrm{new}},\Theta)\,p(\Theta\mid\D)\,d\Theta,
\end{equation}
whose per-table marginal $m_t=\int\mathbb{E}[x_t\mid\phi,G,\Theta]\,p(\Theta\mid\D)\,d\Theta$ we report alongside
it. Parameter and between-database uncertainty are integrated out, not plugged in. The plug-in field
methods compared later instead hold $\theta$ at a cross-validated logistic fit and estimate a scalar evidence
multiplier $\gamma$, an identifiable point-estimate approximation of the same model.

The integrals in \eqref{eq:predictive} are the only quantities we approximate. We fit the maximum a posteriori
$\Theta$ by gradient ascent and form a Laplace posterior $\N(\hat\Theta, H^{-1})$ from the Hessian of the
log-posterior. Conditional on each draw of $\Theta$ the inner expectation in \eqref{eq:predictive} is again an
exact enumeration, so the predictive preserves the computational advantage, and averaging over a handful of
Laplace draws approximates the outer integral over $\Theta$. For a database with no training queries we distinguish a
population-level plug-in prediction, obtained by setting $z_d=0$, from a full predictive calculation that
integrates over $z_d\sim\N(0,I)$, comparing both in Section
\ref{sec:experiments}. Alongside the joint posterior we also apply a greedy conditional-threshold reading of
\eqref{eq:mrf} for visualisation. Starting from the empty selection, it repeatedly commits the uncommitted table of highest full
conditional probability, updates the conditionals of its neighbors through \eqref{eq:mrf}, and stops when no
uncommitted table exceeds one half, never removing a committed table and breaking ties by unary score. Because
the coupling only raises a table's conditional as neighbors commit, this add-only procedure terminates at a
stable selection and yields the belief trace of the trace plot in \hyperref[s:app:extra]{Supplement~\ref*{s:app:extra}} and a variable-size stopping rule.

\paragraph*{Implementation} Inclusion probabilities are computed by exact enumeration over all $2^n$
subsets for $n\le14$, the parameter posterior by a Laplace approximation with $200$ posterior draws, a count validated in \hyperref[s:app:checks]{Supplement~\ref*{s:app:checks}}, and the value
features from a $500$-row scan per table. Feature construction, prior scales and the Hamiltonian Monte Carlo
check are in \hyperref[s:app:checks]{Supplement~\ref*{s:app:checks}}.

\section{Recovery guarantees}\label{sec:theory}

The following result characterizes the couplings for which the gold set is a single-flip local optimum. A
stable selection is a single-flip local optimum of \eqref{eq:mrf}, meaning flipping any one inclusion indicator, from include to exclude or vice versa, worsens the likelihood. This is a local-stability result. Proposition~\ref{prop:global} gives the corresponding global-mode conditions and bounds the posterior mass on incorrect selections, and both proofs are in \hyperref[s:app:proofs]{Supplement~\ref*{s:app:proofs}}.

Write $\kappa=\logit(1/2)=0$ as the cut-off for inclusion based on log-odds. Partition the relevant set $\R$ into salient tables (unary log-odds above
threshold, kept by the independent selector) and weakly-signaled tables (below threshold, missed). For a
weakly-signaled table $b$ let $\delta_b>0$ be its deficit and $k_b\ge1$ its number of relevant neighbors. For
an irrelevant table $r$ let $\mu_r>0$ be its margin and $k'_r$ its number of relevant neighbors. The terms
$\alpha$ and $\gamma$ are absorbed into the $a_t$.

\begin{proposition}[Local stability]\label{prop:window}
The true selection $\R$ is stable if and only if
\[
\max_{b}\frac{\delta_b}{k_b}\;<\;\beta\;<\;\min_{r:\,k'_r\ge1}\frac{\mu_r}{k'_r}.
\]
Such a $\beta$ exists if and only if $\max_b \delta_b/k_b<\min_r \mu_r/k'_r$. The independent selector
($\beta=0$) excludes every weakly-signaled relevant table, so whenever one is present it cannot be recovered
without coupling.
\end{proposition}

The two bounds have a simple reading. The lower bound is set by the relevant table with the largest unary deficit
per relevant neighbor ($\max_b\delta_b/k_b$), the coupling needed to rescue the hardest weak table. The upper
bound is set by the irrelevant table with the smallest margin per relevant neighbor ($\min_r\mu_r/k'_r$), above
which the most vulnerable distractor enters. A useful coupling exists only when the former is below the latter.
For the junction table of Figure \ref{fig:concept}, with
unary probability $0.24$ so deficit $\delta\approx1.15$ and $k=2$ relevant neighbors, the lower bound is
$\delta/k\approx0.58$, while a distractor with margin $\mu=1$ and one spurious neighbor is admitted only above
$\beta=1$, so any coupling between about $0.58$ and $1$ recovers the junction without the distractor.

The lower bound is vacuous when there is no weakly-signaled table, and any $\beta>0$ can only
risk over-selection. If such a table has no relevant neighbor ($k_b=0$) no finite coupling recovers it. The
window can also be empty. When $\max_b\delta_b/k_b\ge\min_r\mu_r/k'_r$, any coupling strong enough to rescue the
hardest weak table has already admitted some distractor. Finally, even in a non-empty window the proposition only
says $\R$ is stable, a fixed point of the conditional-threshold rule. It does not say that a greedy procedure
started from the salient tables reaches $\R$, since two weak tables can support each other once both are in yet
leave neither able to enter first. The result also assumes each distractor starts below threshold. A distractor
the unary model already accepts ($a_r>0$) cannot be corrected by coupling. It is included on its own unary
evidence, regardless of the coupling.

Proposition \ref{prop:window} establishes only single-coordinate stability. A configuration can pass that test and still be beaten by a coordinated move that changes
several tables at once, because a group of distractors can reward one another through the edges among them. The
next result gives conditions for $\R$ to be the global posterior mode, beating every competitor $S$, which we
call global stability in contrast to the local stability of Proposition \ref{prop:window}.

\begingroup\color{blue}
\begin{proposition}[Global mode and a posterior-mass bound]\label{prop:global}
With the notation of Proposition~\ref{prop:window}, the gold set $\R$ is the global posterior mode if and
only if, for every set $A$ of irrelevant tables and every subset $D\subseteq\R$,
\begin{equation}\label{eq:noover}
 \textstyle\sum_{r\in A}(a_r+\beta k_r)+\beta\,e(A)\;\le\;0
\end{equation}
\begin{equation}\label{eq:nounder}
 \textstyle\sum_{t\in D}(a_t+\beta k_t)-\beta\,e(D)\;\ge\;0 ,
\end{equation}
where $e(\cdot)$ counts internal edges and $k_u$ is the number of neighbors of table $u$ lying in $\R$,
written $k'_r$ for an irrelevant table in Proposition~\ref{prop:window}. If in addition the two brackets are
bounded away from zero by $\varepsilon$ per table, then $p(\R)\ge(1+e^{-\varepsilon})^{-|V|}$, so the
posterior cannot spread much mass over incorrect selections. Condition~\eqref{eq:noover} is stronger than the
single-flip bound of Proposition~\ref{prop:window} because of the internal term $\beta e(A)$: a connected
clump of distractors can be jointly favorable even when no single one is.
\end{proposition}
\endgroup

These conditions are the group-level analogues of the bounds in Proposition \ref{prop:window}. The mass bound
adds that when they hold with room to spare, every wrong table costs at least $\varepsilon$ of log-posterior
score, so the posterior cannot spread much of its mass over incorrect selections.

\noindent Proofs of both propositions are in \hyperref[s:app:proofs]{Supplement~\ref*{s:app:proofs}}.

\subsection{Why coupling needs a compensating intercept}\label{sec:intercept}

Proposition~\ref{prop:window} concerns which tables are selected. The last result concerns
the marginal inclusion probabilities, which positive coupling changes even when the unary field is calibrated. By
rewarding every pair of selected neighbors, positive coupling pushes the model toward selecting more tables and
raises the inclusion probability of any table that has selected neighbors. If the unary field was already calibrated, so that $\sigma(a_t)$ matched the rate at which
such tables are truly relevant, then adding coupling on top of it inflates these probabilities above what the
evidence supports, and the marginal inclusion probabilities become overconfident. Because the method is evaluated probabilistically,
this shift affects calibration and not only the ranking. A downstream system reads the marginals as the chance
that a table is needed, so if they are systematically too high it will keep tables it should drop and abstain too
rarely. Holding the expected selection size fixed therefore requires lowering the intercept
$\alpha$, which charges a cost for every selected table and undoes the inflation. The next result makes this
precise. Write $S(x)=\sum_t x_t$ for the number of selected
tables and $T(x)=\sum_{(s,t)\in E}x_s x_t$ for the number of selected adjacent pairs, the two sufficient
statistics paired with $\alpha$ and $\beta$ in \eqref{eq:mrf}.

\begin{proposition}[Mean shift and intercept compensation]\label{prop:intercept}
For the ferromagnetic model \eqref{eq:mrf} with $\beta\ge0$,
\[
\frac{\partial\,\mathbb{E}[S]}{\partial\beta}=\operatorname{Cov}(S,T)\ge 0,
\qquad\text{and}\qquad
\frac{d\alpha}{d\beta}\Big|_{\mathbb{E}[S]\ \mathrm{fixed}}=-\,\frac{\operatorname{Cov}(S,T)}{\operatorname{Var}(S)}\le 0 .
\]
That is, increasing the coupling increases the expected number of selected tables, and holding the expected
selection size fixed requires the intercept to decrease. Moreover $\partial\,\mathbb{E}[x_t]/\partial\beta=
\operatorname{Cov}(x_t,T)\ge0$ for every table $t$, so each individual marginal inclusion probability is
non-decreasing in the coupling.
\end{proposition}

The proposition shows that, at a fixed intercept, increasing $\beta$ shifts the marginal inclusion probabilities
upward. Temperature scaling multiplies the logits and preserves the zero-logit boundary, whereas the required
correction is an additive intercept shift, so temperature cannot in general reproduce it, and on our data
it fails to remove the resulting overconfidence. The correction is a decrease in the intercept, which is exactly
the negative intercept the fit recovers. On our data the identity holds to numerical precision
(the largest gap between $\partial\mathbb{E}[S]/\partial\beta$ and $\operatorname{Cov}(S,T)$ across graphs is
about $10^{-8}$), the intercept that holds the expected selection size fixed decreases monotonically (from about
$0$ at $\beta=0$ to about $-2.2$ at $\beta=2$ on BIRD), and the fitted intercept is negative ($-0.33$ on BIRD,
$-0.60$ on Spider). We refer to this effect as coupling-induced mean shift with intercept compensation, and use
the simulation study only to show that the pattern persists across the simulated graph regimes, not to
define it.

\section{Empirical analysis}\label{sec:experiments}

This section gives the data, protocols and estimators (\S\ref{sec:exp-data}), asks whether coupling recovers
the right tables (\S\ref{sec:exp-recovery}), examines the calibration of the joint posterior
(\S\ref{sec:exp-calib}), and tests how far the findings carry (\S\ref{sec:sim}).

\subsection{Data, estimators, and metrics}\label{sec:exp-data}

\subsubsection{Data and evaluation protocol}

Our two datasets are the standard cross-domain text-to-SQL benchmarks. Spider \citep{yu2018spider} contains
roughly ten thousand human-written natural-language questions paired with reference SQL queries over two hundred
relational databases spanning 138 domains, with database-level train/test separation so that databases seen
during training do not reappear at test time. BIRD \citep{li2023bird} is larger, with over twelve thousand
question--SQL pairs over ninety-five real-world databases drawn from 37 professional domains such as healthcare,
finance, and sports. Its schemas are generally bigger and its joins more complex, so schema selection matters
more there. Each example pairs a natural-language question with a reference (\emph{gold}) SQL query that an expert wrote to answer it against one database.
We take the \emph{gold table set} of an example to be the tables the reference query reads from, parsed from its
\texttt{FROM} and \texttt{JOIN} clauses, and this set is the target our model tries to recover. Recovery and containment therefore measure
agreement with the table set of one reference query, not SQL execution accuracy and not minimal
sufficiency of the schema shown. We keep the
multi-table questions, those whose gold set spans at least two tables and whose schema is small enough for exact
enumeration. This leaves 639 questions over 8 databases from BIRD and 459 over 20 databases from Spider.

The candidate set for a query is the entire schema: every table returned by the
database catalogue is a node, exact enumeration runs over all $2^n$ subsets of it, and nothing is pre-filtered
by connectivity. This matters for interpreting what follows, because a reader could reasonably suspect that
restricting candidates to a connected component would let the foreign-key graph do double duty, once as a filter
and once as evidence. It does not.

It also fixes how much room there is to select. Spider's databases have a median of four tables against a median
gold set of two, so half the schema is the answer at the median and in $11.0\%$ of Spider queries the gold set is
the whole schema. BIRD has a median of eight tables against the same gold size, and the gold set is the
whole schema in $0.8\%$ of queries. Foreign-key edges are present for almost all tables ($3.5\%$ of Spider tables
and $9.2\%$ of BIRD tables have none), and the gold set is connected in the foreign-key graph for $92.6\%$ of Spider and $95.1\%$ of
BIRD queries, with a mean of $1.27$ and $1.38$ internal edges. So the structure the model
uses is real, but on Spider the selection problem it is asked to solve is close to degenerate, which is a second
and more mundane reason, alongside the unary-strength study in the supplement, that the coupling
buys nothing there. We report Spider because it is standard, not because it discriminates.

The \emph{database-seen} (DS) cross-fit splits queries into two folds at random, so every database appears
in both. It measures what the model is worth on a database from which training queries have already accumulated,
which is the adaptation question. The \emph{leave-one-database-out} (LODO) protocol trains the unary weights and the
field on other databases entirely and predicts a database never seen. It measures deployment. Both occur in practice: a database may have already served queries and be asked
more, or it may be newly added to the system. Analyses~\ref{exp:lodo}, \ref{exp:contain}, \ref{exp:card} and~\ref{exp:uncertain} are computed under
LODO in full: the unary logistic weights, the feature standardization, the coupling and, where present, the size
potential and edge model are every one of them fitted on other databases and applied to a database never seen.
For each fold in the DS cross-fit, the unary logistic weights, the feature standardization, and the coupling $\beta$ are all fit on the other
fold, with $\beta$ chosen there to maximize recall. The held-out fold is used only for evaluation, so no query is
scored with parameters fit on itself. We distinguish two categories of hard-to-select gold table. A weakly-signaled
table is one the independent score misses, ranked below the gold set size and scored below one half. It is missed
by construction, and coupling's task is to recover it. A structural connector is a gold table that is a
cut vertex of the gold-induced subgraph, so that removing it disconnects the remaining gold tables, the junction
role of Figure \ref{fig:concept}. The two categories overlap only partly, since a connector may be individually
salient and a weakly-signaled endpoint need not be a connector, and we report recovery for each. Across the
tables that follow, recall and recovery use the top-$|\R|$ ranking, calibration and the proper scores (Brier and
log) use the marginal inclusion probabilities, and set precision uses the $m_t>\tfrac12$ selection. Calibration is
measured by the expected calibration error, which captures how far the stated probabilities sit from the
frequencies they should match. Recall at the gold set size assumes we already know how many tables to keep ($|\R|$), so it measures ranking
quality rather than a rule one could deploy, while the threshold precision evaluates the model's own
variable-size selection.

\subsubsection{Estimators}\label{sec:exp-methods}

The methods below share the same scoring machinery and differ only in how much of the field they estimate and
whether they let it vary by database. A fixed unary is a cross-validated logistic fit on the query--table
features, held constant while the field is estimated. An estimated unary gives the weights $\theta$ a
prior and infers them jointly. We use these names throughout.
\begin{itemize}
\item \textbf{Cosine}: the raw query--table cosine similarity, with no learned weights and no coupling. A pure
similarity baseline.
\item \textbf{Independent}: a fixed unary logistic score with the coupling switched off ($\beta=0$, and the scale
fixed at $\alpha=0$, $\gamma=1$). Tables are scored one at a time, as in standard schema linking, and this is our
main baseline.
\item \textbf{Shortest-path}: a deterministic path-finding selector in the style of SchemaGraphSQL, which takes
the confidently scored tables and adds the tables lying on shortest foreign-key paths between them. It uses the
graph but returns no probabilities.
\item \textbf{Coupled}: the Independent unary fit with a single global coupling $\beta$ turned on, chosen to
maximize held-out recall. The simplest model that can recover weakly-signaled tables, but its coupling is not
itself fitted to the gold sets.
\item \textbf{Global fit}: the global field $(\alpha,\beta)$ estimated together by maximum likelihood on the gold
sets, so the intercept is now free to compensate for the coupling.
\item \textbf{Hierarchical MAP}: the same field with per-database intercepts and couplings, fitted as posterior
modes under a prior that pools them toward the global values, then plugged into the likelihood.
\item \textbf{Hierarchical Bayes}: the full model of Section~\ref{sec:hier}, with $\theta$ given a prior and
inferred jointly with database-level random effects on $\alpha$ and $\beta$. Predictions integrate a Laplace
approximation to the parameter posterior, and this is the method we advocate.
\end{itemize}
Unless stated otherwise, reported results are the \textbf{Hierarchical MAP} fit.
Two analyses differ and we flag them here rather than leave the reader to infer it. The
deployment containment study of Analysis~\ref{exp:contain} and the warehouse study of
\S\ref{sec:beaver} both use a fixed cross-fitted unary field with a single coupling selected on
training databases, that is the \textbf{Coupled} arm rather than the hierarchical fit. The
hierarchical posterior is what the recovery, calibration and predictive-set analyses evaluate. The
deployment gains are therefore attributable to the graph coupling, not to the hierarchy.

\subsection{Recovery: does coupling find the right tables?}\label{sec:exp-recovery}

\expt{Weakly-signaled recovery}{exp:recovery}
Table \ref{tab:selectors} compares the selectors on how well they find the weakly-signaled tables. The coupled
field and the fitted global and hierarchical versions all improve recall over the independent selector
and recover a substantial fraction of the weakly-signaled tables it misses (up to 33\% on BIRD and 36\% on Spider). A deterministic shortest-path baseline in the style of SchemaGraphSQL (Table \ref{tab:selectors}) trades recall
for precision. It has the lowest recall of any method and recovers few weakly-signaled tables, at
precision comparable to the independent selector and well above the coupled posterior. It is a precision-oriented
method that does not match the posterior on recall and provides no calibrated probabilities.

\begin{table}[t]\centering
\small
\caption{Selection quality, cross-fit, computed by a single pipeline (the field methods as in
Table~\ref{tab:calib}). Recall and weakly-signaled recovery at the top-$|\R|$ ranking.
Precision is of the $m_t>\tfrac12$ selection, and connector recovery (Conn.) is the fraction of gold articulation
vertices in the top-$|\R|$ ranking. For the hierarchical model, point-selection metrics use the hierarchical MAP
fit, and the full posterior predictive is evaluated as a probability distribution in Table~\ref{tab:calib}.}\label{tab:selectors}
\resizebox{\linewidth}{!}{%
\begin{tabular}{lcccc@{\hskip 1.5em}cccc}
\toprule
& \multicolumn{4}{c}{BIRD} & \multicolumn{4}{c}{Spider}\\
\cmidrule(lr){2-5}\cmidrule(lr){6-9}
Selector & Recall & Recov. & Prec. & Conn. & Recall & Recov. & Prec. & Conn.\\
\midrule
Independent ($\beta{=}0$)  & 0.782 & 0\%  & \textbf{0.851} & 71\% & 0.898 & 0\%  & 0.880 & 99\%\\
Shortest-path              & 0.665 & 6\%  & 0.843 & 50\% & 0.837 & 16\% & 0.874 & 68\%\\
Coupled, fixed $\beta$     & 0.805 & \textbf{33\%} & 0.726 & 84\% & 0.911 & 20\% & 0.799 & 100\%\\
Global fit      & 0.806 & 16\% & 0.846 & 79\% & 0.918 & 18\% & 0.871 & 100\%\\
Hierarchical MAP     & \textbf{0.822} & 31\% & 0.824 & \textbf{87\%} & \textbf{0.930} & \textbf{36\%} & \textbf{0.888} & 100\%\\
\bottomrule
\end{tabular}}
\end{table}

\subsection{Calibration and the joint distribution}\label{sec:exp-calib}

\expt{Calibration and proper scores}{exp:calib}
Table \ref{tab:calib} reports marginal calibration in the DS setting. The raw coupled
posterior at a fixed coupling is badly overconfident (calibration error 0.142 on BIRD, 0.109 on Spider, against
0.024 and 0.040 for the independent selector). The reliability curves in Figure \ref{fig:reliability} cross
below the diagonal, so the probabilities are systematically too high, not merely too sharp. Temperature scaling
rescales the logits by a common factor and preserves the zero-logit boundary. It therefore cannot reproduce the
additive intercept shift the bias requires. Temperature scaling confirms this, leaving the error essentially unchanged (0.142 to 0.162 on BIRD,
\hyperref[s:app:extra]{Supplement~\ref*{s:app:extra}}), whereas Proposition \ref{prop:intercept} identifies the intercept as the correct additive correction, a constant subtracted from every
score. The fitted global field applies this correction, reducing the calibration error to the independent level
(0.022 on BIRD, 0.023 on Spider) while retaining the recall gain, and the fitted intercept comes out negative.
The hierarchical Bayesian posterior predictive is likewise well calibrated (error 0.025 and 0.034) while
integrating parameter uncertainty. Post-hoc Platt scaling reaches similar calibration but is an external
correction that touches only the marginals (\hyperref[s:app:extra]{Supplement~\ref*{s:app:extra}}).

\begin{table}[t]\centering
\caption{Marginal calibration (expected calibration error), DS. Lower is better. Recall
and weakly-signaled recovery for these methods are in Table \ref{tab:selectors}.}\label{tab:calib}
\begin{tabular}{lcc}
\toprule
Method & BIRD & Spider\\
\midrule
Independent ($\beta{=}0$)            & 0.024 & 0.040\\
Coupled, fixed $\beta$               & 0.142 & 0.109\\
Global fit                & 0.022 & \textbf{0.023}\\
Hierarchical MAP               & \textbf{0.015} & 0.029\\
Hierarchical Bayes (predictive)      & 0.025 & 0.034\\
\bottomrule
\end{tabular}
\end{table}

\begin{figure}[t]\centering
\includegraphics[width=0.6\linewidth]{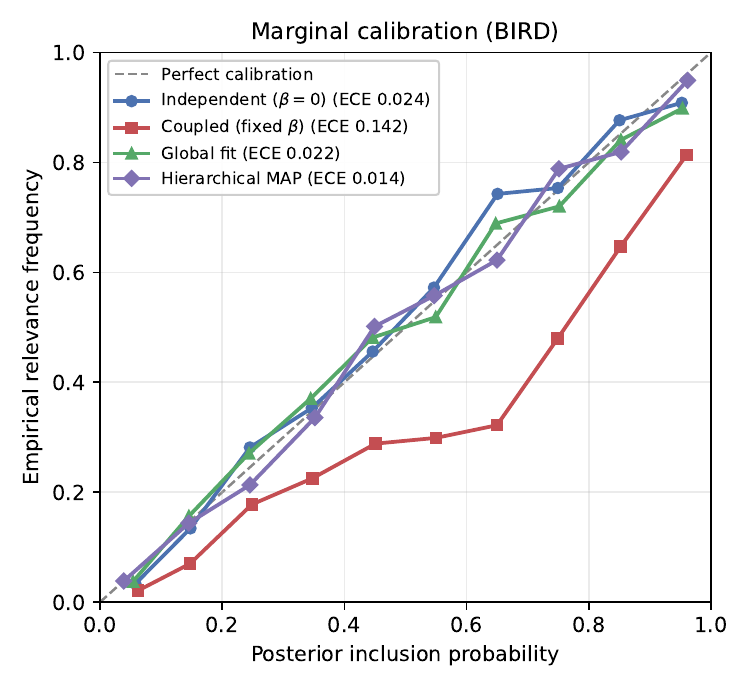}

\caption{Marginal reliability on BIRD. The coupled field at a fixed coupling falls below the diagonal, so it is overconfident and biased, not merely sharp, while the global and hierarchical MAP fits lie on it. The coupling sweep tracing the recovery window of Proposition~\ref{prop:window} is in the supplement.}\label{fig:reliability}
\end{figure}

Two proper scoring rules, the Brier and log scores, tell the same story and are reported in the
supplement with $95\%$ intervals from a bootstrap resampling whole databases, since queries in the same
database are not independent. The hierarchical MAP fit has the lowest Brier and log score on both datasets
with the full posterior predictive nearly tied, so fitting improves overall probabilistic accuracy and not
only calibration. The clustered intervals are wide, since BIRD has only $8$ databases and Spider $20$. The calibration-error
difference between the independent selector and the fitted methods is therefore not individually resolved,
though the intervals do separate the fixed-coupling model from the fitted ones. Evaluated as a distribution over whole subgraphs rather than table by table, the coupled
posterior assigns higher log posterior-predictive probability to the exact gold subgraph than the
independent model, and its $90\%$ predictive sets have near-nominal empirical coverage where the independent
model's undercover. That comparison is the one the method is built for and it is given in full in the
supplement.

\expt{The joint posterior predictive}{exp:joint}
The metrics so far are marginal, table by table. The method's actual output is a distribution over whole
subgraphs, so we evaluate that distribution directly. Table \ref{tab:joint} reports the log
posterior-predictive probability of the exact gold subgraph, the exact-subgraph maximum a posteriori accuracy,
the expected Hamming loss, and the size and coverage of the 90\% posterior predictive set, for the coupled model
as a hierarchical MAP and as the full posterior predictive, against the same model with the coupling forced to
zero.

Most of the improvement here comes from coupling. Both coupled versions assign much higher probability to the
exact gold subgraph than the $\beta=0$ model and recover it exactly far more often (MAP accuracy 23\% against
16\% on BIRD, 53\% against 45\% on Spider), the joint advantage the marginal proper scores did not reveal.
Integrating the parameter posterior adds little to these point metrics. The full predictive matches the
hierarchical MAP on log-probability and Hamming loss without beating it, so on the joint distribution, as on
the marginals, its value is not sharper point prediction.

The clearest difference is in predictive-set coverage. We form a \emph{posterior predictive set} for each query by listing the
candidate table subsets from most to least probable and keeping them until their probabilities reach $0.90$. The
$\beta=0$ model is overconfident on the joint, covering only 0.82 of gold subgraphs on BIRD and 0.89 on Spider at
a nominal 0.90. Coupling restores near-nominal coverage (0.90 to 0.92) with informative sets, mean size about 40
configurations on BIRD and 6 on Spider, far below the $2^{|V|}$ total. The full predictive is the most
conservative, with slightly larger sets and the highest coverage. Good marginal calibration therefore does not
guarantee nominal coverage of complete-subgraph sets. The coupled model restores near-nominal joint coverage,
while integrating the parameter posterior gives the most conservative predictive sets.

\begin{table}[t]\centering
\small
\setlength{\tabcolsep}{4pt}
\caption{Direct evaluation of the joint distribution: log posterior-predictive probability of the exact gold
subgraph, exact-subgraph MAP accuracy (MAP), expected Hamming loss (Ham.), and the mean size and coverage (cov.)
of the 90\% posterior predictive set. All three rows use the DS hierarchical model. The independent row
forces $\beta=0$, the MAP row plugs in the posterior mode, and the predictive row integrates the parameter
posterior.}\label{tab:joint}
\resizebox{\linewidth}{!}{%
\begin{tabular}{lccccc@{\hskip 1em}ccccc}
\toprule
& \multicolumn{5}{c}{BIRD} & \multicolumn{5}{c}{Spider}\\
\cmidrule(lr){2-6}\cmidrule(lr){7-11}
Method & log-pred & MAP & Ham. & set & cov. & log-pred & MAP & Ham. & set & cov.\\
\midrule
Independent ($\beta{=}0$)  & $-3.15$ & 16.0\% & 1.75 & 26 & 0.82 & $-1.79$ & 45.3\% & 1.10 & 6 & 0.89\\
Hierarchical MAP           & $\mathbf{-2.73}$ & \textbf{23.0\%} & \textbf{1.75} & 38 & 0.90 & $\mathbf{-1.51}$ & \textbf{52.7\%} & \textbf{0.95} & 6 & 0.90\\
Hierarchical Bayes (pred.) & $-2.74$ & 21.9\% & 1.77 & 40 & 0.91 & $-1.55$ & 50.5\% & 0.98 & 6 & 0.92\\
\bottomrule
\end{tabular}}
\end{table}

\subsection{Generality and robustness}\label{sec:sim}

\expt{Uncertainty over how much to trust each graph}{exp:coupling}
The fitted model reports a posterior for the coupling in each database, not one value for all. Those
posteriors differ widely, from schemas whose foreign keys carry most of the signal to schemas where the
coupling concentrates near zero and the model declines the graph on its own. The intervals widen where a
database contributes fewer queries, which is the behavior partial pooling should produce. The per-database
posteriors are plotted in \hyperref[s:sec:sim]{Supplement~\ref*{s:sec:sim}}.

\expt{Unseen databases}{exp:lodo}
Under the LODO protocol, in which the coupling and the unary field are both fitted on other databases and
applied to one never seen, the structural benefit transfers but only partly, and better on Spider, which
contributes more training databases. Full posterior integration is worse than the hierarchical MAP fit
here, because the right-skewed posterior for a positive coupling pushes inclusion probabilities upward. The
table is in the supplement.

\expt{Deployable selection: containment against context size}{exp:contain}
Recall at top-$|R|$, used above and throughout the schema-linking literature, tells the selector how many tables
the query needs. Nothing does so at deployment. We therefore re-evaluate with two rules that are given no such
information, a variable-size rule keeping every table with marginal at least $\tau$ and a fixed budget keeping
the $k$ highest marginals, and score them by what a downstream system actually requires. A query missing even one
gold table cannot be answered, so the outcome is containment, the fraction of queries whose entire gold set
is selected, and the cost is the context retained, in tables and in schema tokens. Everything here is
LODO. Table~\ref{tab:contain} reports the fixed-budget comparison, which holds context cost
exactly equal between arms.

The coupling buys containment where the schemas are large and the budget is tight: on BIRD, $+0.036$ at $k=3$ and
$+0.048$ at $k=4$, fading to zero by $k=6$. On Spider it buys nothing, and costs a little at the
tightest budgets, $-0.013$ at $k=2$ and $-0.011$ at $k=3$. Under the variable-size rule, with context matched by interpolating the independent arm to the coupled
arm's mean table count, the coupling shows no advantage on either dataset (between $-0.004$ and $-0.066$). We
report both because they disagree and because the disagreement is interpretable: the fixed budget isolates the
\emph{ranking} induced by the coupling at exactly equal cost, while the threshold rule additionally charges the
rule for the calibration of its probabilities, and Analysis~\ref{exp:calib} has already shown that coupling
inflates marginals before the intercept compensates. So the structural posterior improves the ordering of
tables under a tight budget on the harder benchmark, and our earlier recall figures, which supplied the gold set
size, overstated what a deployment would see.

\begin{table}[t]\centering\small
\caption{Gold-set containment at a fixed table budget, coupled minus independent, under LODO, so
the unary weights, the standardization and the coupling are all fitted on other databases. The budget
fixes context cost exactly, so a positive entry is a gain at equal cost.}\label{tab:contain}
\begin{tabular}{lcccccc}
\toprule
& \multicolumn{3}{c}{BIRD} & \multicolumn{3}{c}{Spider}\\
\cmidrule(lr){2-4}\cmidrule(lr){5-7}
Budget $k$ & coupled & independent & gain & coupled & independent & gain\\
\midrule
2 & 0.432 & 0.432 & $\phantom{+}0.000$ & 0.658 & 0.671 & $-0.013$\\
3 & 0.748 & 0.712 & $+0.036$ & 0.932 & 0.943 & $-0.011$\\
4 & 0.862 & 0.814 & $+0.048$ & 0.983 & 0.983 & $\phantom{+}0.000$\\
5 & 0.911 & 0.889 & $+0.022$ & 0.993 & 0.991 & $+0.002$\\
6 & 0.937 & 0.931 & $+0.006$ & 1.000 & 0.998 & $+0.002$\\
\midrule
tokens at $k{=}4$ & \multicolumn{3}{c}{$87$ vs $86$} & \multicolumn{3}{c}{$41$ vs $41$}\\
\bottomrule
\end{tabular}
\end{table}

\paragraph*{Where the gain lives} The figures above pool every query, including the majority on which the
independent selector already retrieves the whole gold set and coupling has nothing to contribute. Stratifying by
whether the query contains a weakly-signaled gold table, the case the model exists to serve, separates two
quite different behaviors (bootstrap intervals over queries, $\ast$ marks exclusion of zero):

\begin{center}\small
\begin{tabular}{lrrrrr}
\toprule
& $n$ & $k=2$ & $k=3$ & $k=4$ & $k=5$\\
\midrule
BIRD, weakly-signaled table  & $249$ & $+0.129^\ast$ & $+0.137^\ast$ & $+0.120^\ast$ & $+0.060^\ast$\\
BIRD, none                   & $390$ & $-0.059$ & $-0.028$ & $-0.013$ & $-0.008$\\
Spider, weakly-signaled table& $48$  & $+0.125^\ast$ & $+0.021$ & $-0.021$ & $\phantom{+}0.000$\\
Spider, none                & $411$ & $\phantom{+}0.000$ & $-0.002$ & $\phantom{+}0.000$ & $\phantom{+}0.000$\\
\bottomrule
\end{tabular}
\end{center}

On the $39\%$ of BIRD queries where a gold table is weakly signaled, coupling adds $0.137$ containment at
$k=3$, nearly four times the pooled figure, and the gain is significant at every budget we report. On the
remaining $61\%$ it costs a little, never significantly. The pooled number is a blend of the two, so it
understates the effect where the mechanism applies and conceals a small price where it does not. Spider's null
is likewise explained: only $10\%$ of its queries contain a weakly-signaled gold table at all, so there is
almost nothing for the graph to do.

This stratification needs the gold set to identify which queries contain a weakly-signaled table, so it is a
diagnostic rather than a deployable gain. It locates the benefit predicted by Section~\ref{sec:coupling} in
exactly the queries that section says should show it, which is stronger evidence for the mechanism than a pooled
average, while the deployable figures remain the ones in Table~\ref{tab:contain}. A rule that could spot the
condition from the unary field alone would turn the $0.137$ into an operational gain, and we have not found one.

A second stratification, this one requiring no oracle, asks where the consequential regime begins.
Splitting queries by the size of their database's schema:

\begin{center}\small
\begin{tabular}{lrrr}
\toprule
Schema (tables) & $n$ & gain at $k=3$ & $95\%$ interval\\
\midrule
BIRD, $\le4$    & $223$ & $-0.022$ & $[-0.045,+0.000]$\\
BIRD, $5$--$6$  & $42$  & $-0.024$ & $[-0.119,+0.048]$\\
BIRD, $7$--$9$  & $182$ & $+0.088$ & $[+0.044,+0.132]$\\
BIRD, $\ge10$   & $192$ & $+0.068$ & $[+0.005,+0.135]$\\
Spider, all     & $459$ & $\phantom{+}0.000$ & ---\\
\bottomrule
\end{tabular}
\end{center}

The gain changes sign at about seven tables, and both larger strata exclude zero while neither smaller one
does. The share of queries containing a weakly-signaled gold table rises with schema size in the same way,
from $25\%$ below five tables to about $50\%$ above seven, which is the mechanism behind the threshold: bigger
schemas give the unary field more chances to miss something the graph can recover. Spider is flat at zero
across its whole range because that range does not reach the threshold. These are real queries under LODO, with no constructed examples, so the claim that structured selection becomes consequential
only at scale is now measured rather than asserted, and the boundary is located.

\expt{When the join graph is not given}{exp:uncertain}
Every result so far conditions on a declared foreign-key graph, which a production warehouse frequently does
not ship: of the $28$ Spider~2.0 databases, none declares a single foreign key. We therefore replace the given
graph with a model of it, fitting a logistic edge model on declared keys from other databases using schema
evidence alone, and weighting the pairwise term by its prediction. Because BIRD and Spider do declare keys, the
inference can be priced against the declared-graph ceiling.

\begin{table}[t]\centering\small
\caption{Containment at a fixed budget with the join graph inferred rather than declared, BIRD, true
LODO throughout, including the edge model. Gains are over the independent
selector. The independent and declared-graph arms are recomputed inside the edge-model pipeline, which selects
$\beta$ on its own grid, so they differ from Table~\ref{tab:contain} in the third decimal.}\label{tab:uncertain}
\begin{tabular}{lcccl}
\toprule
graph & $k{=}2$ & $k{=}3$ & $k{=}4$ & gain at $k{=}3$\\
\midrule
independent (no graph)        & $0.432$ & $0.712$ & $0.814$ & ---\\
declared foreign keys         & $0.432$ & $0.746$ & $0.858$ & $+0.0344$ $[+0.0110,+0.0579]$\\
inferred, weighted            & $0.423$ & $0.732$ & $0.850$ & $+0.0203$ $[-0.0031,+0.0454]$\\
inferred, thresholded         & $0.426$ & $0.739$ & $0.856$ & $+0.0266$ $[+0.0047,+0.0485]$\\
inferred, declared removed    & $0.432$ & $0.712$ & $0.814$ & $+0.0000$\\
\bottomrule
\end{tabular}
\end{table}

The edge model reaches LODO AUC $0.708$ on average and $0.876$ at the median. Thresholded
inferred edges recover $0.59$ of the declared-edge gain and are statistically indistinguishable from the
declared graph, $-0.0141$ $[-0.0329,+0.0047]$, so the method does not require the keys to be shipped. It does
not improve on them either: weighting edges by posterior probability does no better than a hard threshold, and
deleting the declared edges leaves no gain at all, the coupling selecting $\beta=0$ of its own accord. This is
a robustness result: it removes a data-availability assumption and does not add accuracy. The edge features,
the weighted comparison and a transfer attempt to Spider~2.0 are in \hyperref[s:sec:sim]{Supplement~\ref*{s:sec:sim}}.

\expt{Separating how many from which, and which of the two is doing the work}{exp:card}
Analysis~\ref{exp:contain} left a specific failure. Coupling improves ranking at a fixed budget and adds
nothing to a variable-size rule at matched context. The model explains why. In
\[
   \log p(S\mid a,G)=\sum_{t\in S}a_t+\alpha|S|+\beta T(S)-\log Z ,
\]
the sufficient statistic $T(S)$ counts selected edges, and that count grows with the \emph{size} of $S$ as
well as with local density. A scalar attractive coupling therefore changes how many tables are selected and
not only which ones. The intercept $\alpha$ can absorb the average size shift, which is
Prop.~\ref{prop:intercept}, but it cannot absorb a query-specific one, and a variable-size rule is precisely
the rule that bears the consequence.

The remedy is to give size its own term and stop asking the coupling to carry it,
\begin{equation}\label{eq:card}
   \log p(S\mid a,G)=\sum_{t\in S}a_t+\beta T(S)+c(|S|),
\end{equation}
with $c$ a log size distribution fitted on the training databases. It costs nothing: exact enumeration
already groups configurations by size, so $c$ is added to each configuration by its cardinality.

\paragraph*{How $c$ is estimated} $c$ is the log of an add-one-smoothed empirical distribution of gold-set
sizes over the training databases, supported on $1,\dots,14$ with the empty set given negligible mass. It is
fitted on multi-table queries, which are the study population throughout. This is a plug-in potential
rather than a further level of the hierarchical model of \S\ref{sec:hier}: no posterior over $c$ is carried,
and \eqref{eq:card} should be read as a pre-registered extension of the fixed-unary pipeline rather than as
the full Bayesian estimator. We keep that distinction visible because the paper's other calibration claims
depend on it. A variant conditioning $c$ on schema size was also run and is exploratory. It did not separate
from the pooled version and we do not report it further.

\paragraph*{The comparison that decides what the gain is} A size distribution is a useful thing to know
independently of any graph, so the coupled-versus-coupled-plus-$c$ contrast cannot establish that the
graph is contributing. The controlled comparison is against an independent selector given the same
size prior. All four arms are matched on retained context by tuning each arm's inclusion threshold on the
training databases to a common target mean, so every column below spends the same budget.

\begin{table}[t]\centering\small
\caption{Gold-set containment at matched mean retained context, paired over queries, with the unary field,
the coupling and $c$ all fitted under LODO. The final column is the controlled comparison.}
\label{tab:card}
\begin{tabular}{lccccl}
\toprule
target & indep. & indep.$+c$ & coupled & coupled$+c$ & coupled$+c$ minus indep.$+c$\\
\midrule
\multicolumn{6}{l}{\emph{BIRD}}\\
$3.0$ & $0.624$ & $\mathbf{0.720}$ & $0.535$ & $0.645$ & $-0.0751$ $[-0.1049,-0.0454]$\\
$4.0$ & $0.797$ & $\mathbf{0.844}$ & $0.775$ & $0.829$ & $-0.0141$ $[-0.0344,+0.0078]$\\
$5.0$ & $0.905$ & $\mathbf{0.931}$ & $0.905$ & $0.909$ & $-0.0219$ $[-0.0376,-0.0063]$\\
$6.0$ & $0.978$ & $\mathbf{0.981}$ & $0.972$ & $0.977$ & $-0.0047$ $[-0.0110,+0.0000]$\\
\multicolumn{6}{l}{\emph{Spider}}\\
$2.0$ & $0.599$ & $\mathbf{0.697}$ & $0.545$ & $0.612$ & $-0.0850$ $[-0.1198,-0.0523]$\\
$2.5$ & $0.771$ & $\mathbf{0.882}$ & $0.712$ & $0.808$ & $-0.0741$ $[-0.1024,-0.0479]$\\
$3.0$ & $0.874$ & $\mathbf{0.963}$ & $0.854$ & $0.917$ & $-0.0458$ $[-0.0697,-0.0240]$\\
$4.0$ & $0.974$ & $\mathbf{0.998}$ & $0.972$ & $0.998$ & $+0.0000$ $[+0.0000,+0.0000]$\\
\bottomrule
\end{tabular}
\end{table}

Table~\ref{tab:card} confirms the diagnosis and then refutes the reading we first gave it. The diagnosis
holds: at matched context the coupled model is no better than the independent one and on Spider is
\emph{worse}, which is what size inflation looks like at a fixed budget, and adding $c$ repairs it, by
$+0.049$ on BIRD and up to $+0.096$ on Spider. What does not hold is the attribution. The size prior helps
the \textbf{Independent} selector by as much or more, $+0.0955$ at three tables on BIRD and $+0.1111$ at two
and a half on Spider, and once both arms have it the graph contributes nothing: the controlled difference is
negative at six of eight targets and its interval excludes zero at four, never favoring the coupled arm
anywhere. On this frontier the useful quantity is the size distribution, and the topology is at best
redundant and at worst an active distraction.

A reading of this experiment without the controlling arm would credit the coupling with the gain, which is
why the control is reported alongside it.

\paragraph*{Where each term matters} The two regimes separate cleanly.
\begin{itemize}\itemsep1pt
\item When the size is \emph{imposed} from outside, as at a fixed table budget, there is no size decision to
get wrong and topology is what remains to exploit. That is the regime of
Analyses~\ref{exp:recovery}--\ref{exp:contain}, where coupling adds $0.036$ pooled, $0.137$ where a gold
table is weakly signaled, and $0.0775$ on schemas of seven tables or more.
\item When the size is \emph{chosen}, the size decision dominates, and a one-parameter-per-size log
distribution captures it better than an attractive coupling ever could, because the coupling reaches size
only through a statistic that also encodes topology.
\end{itemize}
The practical recommendation follows from the split itself. A deployment that must
choose how many tables to retain should fit a size prior first. It is nearly free and it is worth more than
the graph. A deployment with a fixed context budget should use the graph, and \S\ref{sec:beaver} says at what
schema size that starts to matter.

\section{Scaling beyond exact enumeration}\label{sec:scale}

Exact evaluation of \eqref{eq:post} enumerates all $2^{|V|}$ configurations, which is comfortable for
the schema graphs studied above but not for large ones. This section shows that the model extends past that boundary while remaining exact on a
restricted support, and reports the regime in which selection is worth
performing at all. Both questions turn on the same quantity, the size of the graph, and they answer in
the same direction.

\subsection{Candidate-restricted exact inference}\label{sec:candidate}

Rather than approximate the distribution, we restrict its support. Rank tables by the unary evidence
$a_t$, form a candidate pool $C$ from the top $K$ together with their foreign-key neighbors, and
enumerate exactly over the induced subgraph on $C$ while fixing $x_t=0$ for $t\notin C$. The
neighbor closure is essential rather than cosmetic. A structurally supported table has weak unary
evidence by construction, so a pool chosen on evidence alone would exclude precisely the tables the
coupling exists to recover.

The restriction is validated where both computations are available. On the benchmark schemas, where
full enumeration is feasible, the restricted posterior agrees with the full-exact one from a modest
pool onwards: at $K=6$ the mean absolute difference in inclusion marginals is $2\times10^{-4}$, the
maximum a posteriori subset agrees on every query, and recall ($0.808$) and predictive-set coverage
($0.908$) are identical to the full calculation, using $1{,}090$ of the $2{,}576$ configurations that
full enumeration averages. By $K=8$ the marginal difference is zero to four decimal places. That
agreement is what licenses the restriction where full enumeration is impossible.

\subsection{A 97-table enterprise warehouse}\label{sec:beaver}

We apply the restricted posterior to the \textsc{dw} warehouse of the BEAVER benchmark
\citep{chen2024beaver}: 97 tables joined by $127$ undirected table-level join edges, with $120$ queries touching $4.0$
tables on average. The released join-key file lists $1{,}034$ \emph{column}-level key pairs. These collapse to
$186$ ordered and $127$ undirected table pairs once self-joins are removed, and it is the table-level graph the
model uses. Enumerating $2^{97}$ configurations is not available. Full closure is not available either: the
foreign-key neighborhood of $15$ unary seeds on this graph has a median of $38$ tables, so exact
enumeration over the closure would require $2^{38}$ configurations. We therefore use a \emph{bounded}
closure, taking the top $8$ tables by unary evidence together with their foreign-key neighbors,
admitted in unary order until the pool reaches $16$. This keeps the enumeration exact while allowing
structure, rather than evidence alone, to determine part of the support.

\begin{table}[t]\centering\small
\caption{Recall at the true gold-set size on all three BEAVER warehouses, with the arms of
\S\ref{sec:scale}. The coupled model uses the bounded-closure posterior of \S\ref{sec:candidate}. The best
\emph{unary} baseline changes between databases, so choosing one requires knowing which warehouse one is on.
The coupled posterior does not have to be chosen.}\label{tab:warehouses}
\begin{tabular}{lcccc}
\toprule
arm & \textsc{dw} & \textsc{neutron} & \textsc{nova} & worst case\\
\midrule
cosine                   & $0.425$ & $\mathbf{0.346}$ & $0.212$ & $0.212$\\
unary fusion             & $0.511$ & $0.167$ & $0.446$ & $0.167$\\
foreign-key closure      & $0.511$ & $0.201$ & $0.446$ & $0.201$\\
PageRank diffusion       & $0.550$ & $0.281$ & $\mathbf{0.744}$ & $0.281$\\
coupled, bounded closure & $\mathbf{0.563}$ & $0.341$ & $0.724$ & $\mathbf{0.341}$\\
\bottomrule
\end{tabular}
\end{table}

Table~\ref{tab:warehouses} shows two things happening together. Independent evidence degrades sharply:
table cosine attains 0.425 here against 0.72 on the smaller benchmarks, because 97 correlated tables
sharing name prefixes are far harder to separate on surface features than a handful of small schemas.
Structural methods retain most of their performance, so the gap that structure closes is much wider
here than on BIRD or Spider. This is the regime the model was built for, and it is the regime the
public benchmarks do not exhibit.

The bounded closure does what it is intended to do, and the effect is clearest in pool coverage rather
than in recall. Restricting to the top $15$ tables by unary evidence places the entire relevant set
inside the pool for $44.2\%$ of queries. The bounded closure raises that to $59.2\%$, because a table
whose relevance is structural rather than semantic cannot enter a pool selected on evidence alone.
Recall improves only from $0.556$ to $0.563$. The binding constraint at this scale is therefore
selection within the pool rather than coverage of it.

The coupled model's advantage over PageRank
diffusion is $+0.006$ with interval $[-0.021, +0.033]$, so on recall the two are not separated at this
sample size. What separates them is the object returned rather than the ranking quality. A diffusion
score gives an ordering, whereas the coupled posterior gives a calibrated distribution over subsets,
which is what a downstream decision under asymmetric costs requires and what
Section~\ref{sec:exp-calib} evaluates. The claim here concerns the inferential output, not accuracy.

\paragraph*{Two further warehouses} The same release carries
\textsf{neutron} ($175$ tables, $160$ declared edges, $30$ multi-table questions) and \textsf{nova} ($109$
tables but only $24$ declared edges), and Table~\ref{tab:warehouses} reports all three. One limit has to be
stated first. On \textsf{neutron} the semantic field is weak enough that whole-gold-set containment is out of
reach at any enumerable pool size, so recall at the true gold size is used throughout this subsection. It is
an oracle ranking metric and not the deployable quantity of Analysis~\ref{exp:contain}.

The robustness comparison is more informative than the peak accuracy. Which \emph{unary} baseline is
best flips completely between databases: cosine is far ahead on \textsf{neutron} ($0.346$ against $0.167$)
and far behind on \textsf{nova} ($0.212$ against $0.446$), so choosing a unary arm requires knowing which
warehouse one is on and choosing wrong costs about half the achievable recall. The coupled posterior has the
best worst case of any arm, $0.341$ against $0.281$ for the next best. A worst case over three schemas is a
descriptive minimax carrying no useful interval, so it should be read as a stress test of robustness rather
than a deployment result. Across the three the coupled posterior does not reliably separate from PageRank
diffusion, ahead by $+0.059$ $[+0.005,+0.122]$ on \textsf{neutron}, behind by $-0.020$ $[-0.079,+0.048]$ on
\textsf{nova}, and ahead by $+0.013$ $[-0.013,+0.040]$ on \textsf{dw}.

\section{Discussion}\label{sec:disc}

% AOAS reorder. Paragraph text is unchanged except where marked in blue. The order was
% calibration theory -> subset size -> when selection is useful -> Bayesian role -> limitations.
% It now runs application -> practical rule -> statistical finding -> what the Bayesian
% treatment buys -> limitations, so the section opens and closes on the problem rather than on
% autologistic parameterisation.

\paragraph*{When selection is useful}
Schema selection matters primarily when the full schema consumes substantial context and independent table
scores are insufficient. The schema accounts for $84.9\%$ of the model input in our enterprise setting, while
its median size increases from $312$ tokens on the public benchmarks to $8{,}246$ on \textsc{dw}, a factor of
$26$. Filtering saves little at the smaller scale. At the larger scale, Table~\ref{tab:warehouses} shows recall
from independent evidence falling from $0.72$ to $0.425$, leaving more room for structure to help. We do not
estimate end-to-end cost or latency because both depend on deployment choices outside the scope of this study.

\paragraph*{Subset size and graph structure}
Fit a size prior first, and reach for the graph when an external budget fixes how many tables are
kept. The edge statistic affects both which tables are selected and how many are selected. Analysis~\ref{exp:card}
shows that when the model chooses the subset size, coupling adds no improvement once the independent selector
is given the same size prior.

A production warehouse may not provide a declared foreign-key graph. In Analysis~\ref{exp:uncertain}, inferred
edges recover part of the benefit obtained from declared keys but do not improve on those keys. This relaxes
the data requirement without adding accuracy.

\paragraph*{Calibration under graph coupling}
Graph coupling can recover relevant tables supported by neighboring selections, but it also raises marginal
inclusion probabilities. This upward shift is known. An uncentered attractive autologistic field has a
non-negative autocovariate that pushes realizations toward one, motivating the centered parameterization of
\citet{caragea2009autologistic} and its later development for lattice data
\citep{hughes2011autologistic,wolters2017better}. Proposition~\ref{prop:intercept} identifies the direction of
the intercept adjustment needed to maintain calibrated inclusion probabilities when the target is a joint
distribution over subsets rather than a collection of marginals. It does not guarantee that fitting achieves
calibration, and a centered coding remains an alternative we have not evaluated. In practice, coupling helps
when the schema is large enough for the context budget to bind and the unary scores miss relevant tables that
the graph can recover. These conditions occur more often in BIRD than in Spider. When the graph supports
distractors as strongly as weakly signaled relevant tables, coupling instead leads to over-selection.

\paragraph*{Role of the Bayesian analysis}
The Bayesian formulation provides three outputs a point estimate does not. The first is a joint
predictive distribution over subsets for downstream decisions with asymmetric costs. The second is
database-specific uncertainty about the coupling, which increases when fewer queries are available
(\hyperref[s:sec:sim]{Supplement~\ref*{s:sec:sim}}). The third is predictive sets with near-nominal coverage where the independent model
undercovers. What it does not buy is consistently sharper marginal prediction. The hierarchical MAP fit
ties or outperforms the full posterior predictive, and integration performs worse on an unseen database because
the right-skewed posterior for a positive coupling raises inclusion probabilities.

\paragraph*{Limitations}
\emph{What the evaluation does not establish.} The evaluation does not establish higher end-to-end
execution accuracy, and a deterministic shortest-path baseline attains higher precision at lower recall. We
were unable to turn the results for weakly signaled tables into a deployable gating rule because any such
rule uses the unary field and inherits its optimism, and \hyperref[s:sec:supp-neg]{Supplement~\ref*{s:sec:supp-neg}} reports that attempt alongside the other ablations and withdrawn results.
\emph{What the model assumes.} The autologistic model is conditional on the evidence and graph and
does not specify a generative distribution for the features. Propositions~\ref{prop:window} and~\ref{prop:global} give stability conditions for the gold set, local and global respectively, not a contraction result for a procedure that would reach it.
\emph{What computation and transfer cost.} Exact enumeration is limited to small graphs, so
\S\ref{sec:candidate} instead restricts the support,
with fidelity verified only where full enumeration is feasible, and \S\ref{sec:beaver} applies this approach to
a $97$-table warehouse. The Laplace approximation agrees with Hamiltonian Monte Carlo for the reported
quantities but understates uncertainty in individual database effects. Calibration transfers only partially to
unseen databases and transfers better on Spider, which has more training databases, although the benchmarks
differ in other respects.

Our results establish the calibration behavior for foreign-key-coupled schema
selection and, through simulation, for graph-coupled selection more broadly. A general account of when fitting
restores calibration in coupled selection models remains open. None of this changes the practical
conclusion. Schema selection matters when the schema is large, and there a fitted size prior settles how many
tables to keep while the graph settles which ones.

\begin{acks}[Declaration of generative AI use]
Generative AI tools were used as a coding and drafting assistant, to implement and run analysis scripts, to
draft and revise text, and to check internal consistency. All modeling decisions, pre-registered hypotheses,
failure conditions and interpretations are the authors', who verified every reported number against the named
scripts and take full responsibility for the content. The authors report there are no competing interests to
declare.
\end{acks}

\begin{funding}
No funding was obtained for the reported work.
\end{funding}

\begin{supplement}
\stitle{Code, data and analysis registry}
\sdescription{All scripts named in the text and a \texttt{REPRODUCE.md} giving the command and protocol flag
behind every reported table, and the project page and official repository for each benchmark. BIRD, Spider
and BEAVER are public and are not redistributed. BEAVER is enterprise data and its release is gated, so
repeating the warehouse results needs access approval. The registry
records every modification tested during revision with its outcome, failures and withdrawn results at the
same prominence as successes.}
\end{supplement}

\bibliographystyle{spr-ims-nameyear}
\bibliography{refs}

\clearpage
\setcounter{section}{0}
\setcounter{table}{0}
\setcounter{figure}{0}
\setcounter{equation}{0}
\setcounter{proposition}{0}
\setcounter{experiment}{0}
\renewcommand{\thesection}{S\arabic{section}}
\renewcommand{\thetable}{S\arabic{table}}
\renewcommand{\thefigure}{S\arabic{figure}}
\renewcommand{\theequation}{S\arabic{equation}}
\renewcommand{\theproposition}{S\arabic{proposition}}
\renewcommand{\theexperiment}{S\arabic{experiment}}
\specialsection*{Supplementary Material}

\noindent This supplement collects the proofs, computational checks, and additional results that support the
main paper. References to Propositions~1, 2 and~3 are to the main paper, which states local stability, the global-mode
result and the mean-shift/intercept result. Proposition~2 is restated here for convenience alongside its proof.

\section{Proofs}\label{s:app:proofs}

\begin{proof}[Proof of Proposition~1 of the main paper, local stability]
At $x=\ind{\cdot\in\R}$ the conditional log-odds of table $t$ is $a_t+\beta\,|\N(t)\cap\R|$. A salient table
has $a_t\ge\kappa$ and stays in for all $\beta\ge0$. A weakly-signaled $b$ has conditional log-odds
$\kappa-\delta_b+\beta k_b$, at least $\kappa$ exactly when $\beta\ge\delta_b/k_b$. An irrelevant $r$ has
$\kappa-\mu_r+\beta k'_r$, below $\kappa$ exactly when $\beta<\mu_r/k'_r$ (automatic if $k'_r=0$). Intersecting
gives the window. At $\beta=0$ each weakly-signaled table has log-odds $a_b<\kappa$ and is excluded.
\end{proof}

\begin{proposition}[Global mode and a posterior-mass bound]\label{sprop:map}
With the notation of Proposition~1 of the main paper, the gold set $\R$ is the global posterior mode if and
only if, for every set $A$ of irrelevant tables and every subset $D\subseteq\R$,
\begin{equation}\label{seq:noover}
   \textstyle\sum_{r\in A}(a_r+\beta k_r)+\beta\,e(A)\;\le\;0
\end{equation}
\begin{equation}\label{seq:nounder}
   \textstyle\sum_{t\in D}(a_t+\beta k_t)-\beta\,e(D)\;\ge\;0 ,
\end{equation}
where $e(\cdot)$ counts internal edges and $k_u$ is the number of neighbors of table $u$ lying in $\R$, written $k'_r$ for an irrelevant table in Proposition~1 of the main paper. If in addition the two brackets are bounded away from zero by
$\varepsilon$ per table, then $p(\R)\ge(1+e^{-\varepsilon})^{-|V|}$, so the posterior cannot spread much mass
over incorrect selections. Condition~\eqref{seq:noover} is stronger than the single-flip bound of
Proposition~1 because of the internal term $\beta e(A)$: a connected clump of distractors can be jointly
favorable even when no single one is.
\end{proposition}

\begin{proof}[Proof of Proposition~\ref{sprop:map}]
Write $S=(\R\setminus D)\cup A$ with $D=\R\setminus S$ and $A=S\setminus\R$, and $H(x)=\sum_t a_t x_t+\beta
\sum_{(s,t)\in E}x_s x_t$. Expanding the edge counts gives
\[
H(\R)-H(S)=\Big[\textstyle\sum_{t\in D}(a_t+\beta k_t)-\beta e(D)\Big]-\Big[\textstyle\sum_{r\in A}a_r+\beta e(A)
+\beta\,\mathrm{cut}(A,\R\setminus D)\Big],
\]
where the removal term uses $e(D)+\mathrm{cut}(D,\R\setminus D)=\sum_{t\in D}k_t-e(D)$. The removal term is
nonnegative by \eqref{seq:nounder}. Since $\mathrm{cut}(A,\R\setminus D)\le\mathrm{cut}(A,\R)=\sum_{r\in A}k_r$,
the addition term is at most $\sum_{r\in A}(a_r+\beta k_r)+\beta e(A)\le0$ by \eqref{seq:noover}, so
$H(\R)-H(S)\ge0$. Conversely, if $\R$ is the global mode, comparing it with $\R\cup A$ (the case $D=\varnothing$,
for which $\mathrm{cut}(A,\R\setminus D)=\mathrm{cut}(A,\R)$) gives \eqref{seq:noover}, and comparing it with
$\R\setminus D$ (the case $A=\varnothing$) gives \eqref{seq:nounder}, so the two conditions are also necessary.
Under the set-wise margins the removal bracket is at least $\varepsilon|D|$ and the addition
bracket at most $-\varepsilon|A|$, giving $H(\R)-H(S)\ge\varepsilon\,d_H(\R,S)$. Summing $p(S)/p(\R)=e^{-(H(\R)-H(S))}$ over
Hamming shells, $\sum_{S\ne\R}e^{-\varepsilon d_H}\le\sum_{k\ge1}\binom{|V|}{k}e^{-\varepsilon k}=(1+e^{-\varepsilon})^{|V|}-1$,
and $p(\R)\ge(1+e^{-\varepsilon})^{-|V|}$.
\end{proof}

\begin{proof}[Proof of Proposition~3, intercept compensation]
The main paper's model is an exponential family with natural parameters $(\alpha,\beta)$ (holding
$\gamma,\phi,G$ fixed) and sufficient statistics $(S,T)$. Hence $\mathbb{E}[S]=\partial\log Z/\partial\alpha$
and
\[
\partial\mathbb{E}[S]/\partial\beta=\partial^2\log Z/\partial\alpha\partial\beta=\operatorname{Cov}(S,T),
\]
which is non-negative because $S$ and $T$ are both non-decreasing in $x$ and the attractive (ferromagnetic)
Ising model is positively associated \citep[the FKG inequality,][]{fortuin1971}. Differentiating the constraint
$\mathbb{E}[S]=c$ gives $\operatorname{Var}(S)\,d\alpha+\operatorname{Cov}(S,T)\,d\beta=0$, since
$\partial\mathbb{E}[S]/\partial\alpha=\operatorname{Var}(S)$, and rearranging yields the second identity. The same
argument applied to the indicator $x_t$ in place of $S$ gives
$\partial\mathbb{E}[x_t]/\partial\beta=\operatorname{Cov}(x_t,T)\ge0$, since $x_t$ and $T$ are both non-decreasing
in $x$.
\end{proof}

\section{Computational checks}\label{s:app:checks}

\paragraph*{Validating the Laplace approximation}
The likelihood is exact, but we approximate the posterior over the parameters $\Theta$ by a Gaussian (the Laplace
approximation), so we check it against Hamiltonian Monte Carlo, which targets the posterior without the Gaussian approximation. We compare the
two in coordinates rescaled so that a perfect Laplace approximation would look like a standard normal $\N(0,I)$.
How far the HMC samples stray from that standard normal measures the error directly. We run this check on a
held-out fold of each dataset. The population parameters $(\theta,\alpha_0,b_0,\log\sigma_\alpha,\log\sigma_\beta)$
stay close to the reference, and the aggregate quantities we actually report, recall and marginal calibration,
agree with HMC. On BIRD the calibration error is 0.021 and 0.024 across two independent HMC chains against 0.026
for Laplace and 0.020 for the maximum a posteriori estimate. On Spider it is 0.033 and 0.031 against 0.044 for
Laplace, with recall identical to two decimal places in every case. The one place the approximation is inaccurate
is the individual database random effects, whose true posterior is wider and shifted than the Gaussian allows
(the funnel-shaped posterior typical of hierarchical models). This did not materially change the aggregate
predictive marginals in the folds examined, but we flag it as a limitation of the parameter-level uncertainty.
We report results with 200 posterior draws, and increasing the count further produced negligible changes.

\paragraph*{Prior sensitivity}
Because the hierarchy is central and there are only 8 databases in BIRD, we check that the reported quantities do
not depend on the prior scales. Scaling every prior standard deviation by $\tfrac12$ and by $2$ (a four-fold
range) moves the posterior-predictive calibration error by at most 0.007, the log score by at most 0.005, and
recall by at most 0.003 on either dataset. The prior is broad relative to the fitted quantities rather than
tightly concentrated. Its implied coupling $\beta_d$ has median $0.65$ with a 90\% interval of $[0.01,5.1]$,
comfortably covering the fitted values, while the fitted posterior over selected-set sizes is far more
concentrated than the prior predictive.

\section{Additional results}\label{s:app:extra}

\paragraph*{Budget-matched decision comparison}
Table~\ref{tab:budget} gives the mean number of tables each rule must pass to reach a target retain-all rate. At a matched budget the joint containment set is competitive with a tuned top-$k$ rule, modestly better on Spider and roughly even on BIRD.

\begin{table}[h]\centering
\caption{Budget-matched decision comparison: the mean number of tables that must be passed to reach a target retain-all rate, out of sample. Point rules are swept over their size $k$ and containment sets over $\eta$. At a matched budget the joint containment set is competitive with a tuned top-$k$ rule while additionally supplying a continuous, calibrated containment set whose empirical coverage is near nominal in these experiments, which a top-$k$ rule does not provide. We claim no finite-sample coverage theorem for it.}\label{tab:budget}
\begin{tabular}{lccc@{\hskip 2em}ccc}
\toprule
& \multicolumn{3}{c}{BIRD} & \multicolumn{3}{c}{Spider}\\
\cmidrule(lr){2-4}\cmidrule(lr){5-7}
Rule (retain-all target) & 0.90 & 0.95 & 0.975 & 0.90 & 0.95 & 0.975\\
\midrule
Containment $S_\eta$ (joint)        & 4.29 & 5.24 & 6.04 & 2.68 & 3.07 & 3.45\\
Containment (indep.\ $\beta{=}0$)   & 4.45 & 5.24 & 6.14 & 2.61 & 3.11 & 3.45\\
Top-$k$ (coupled marginal)          & 4.42 & 5.07 & 6.24 & 2.94 & 3.52 & 3.52\\
Top-$k$ (unary score)               & 4.42 & 5.66 & 6.24 & 2.94 & 2.94 & 3.52\\
\bottomrule
\end{tabular}
\end{table}

\paragraph*{Split-conformal containment}
The containment guarantee assumes calibration. When it cannot be assumed, the level can instead be set by split-conformal calibration \citep{lei2018conformal}. On a held-out fold we take for each query the smallest level at which its gold set is retained, $u_i=\min\{\eta:\R_i\subseteq S_\eta\}$, and set $\eta$ to the $\lceil(n+1)(1-\alpha)\rceil$-th smallest of the $u_i$, which controls the retain-all rate at $1-\alpha$ for an exchangeable test query in finite samples. This works, but on schemas this small it is conservative. One global level has to cover the hardest queries,
the ones whose required tables the posterior ranks last, so the conformal set grows to take in almost the whole
schema: a mean of $7.4$ tables out of $7.4$ on BIRD at target $0.90$, against $5.9$ for the calibrated
containment set. It does retain every gold set. So the distribution-free guarantee is a fallback. What buys an
efficient decision here is the calibration the joint fit restores, and that matters more on larger schemas,
where returning the whole schema is not an option.

\paragraph*{Post-hoc recalibration}
Table~\ref{tab:posthoc} gives two post-hoc corrections applied to the fixed-coupling marginals. Temperature
scaling leaves the calibration error essentially unchanged, consistent with Proposition~3 of the main paper,
since it rescales the logits rather than shifting them. Platt scaling reaches calibration comparable to the
fitted models but is an external correction that touches only the marginals.

\begin{table}[h]\centering
\caption{Post-hoc recalibration of the fixed-coupling marginals (expected calibration error, cross-fit); lower is
better.}\label{tab:posthoc}
\begin{tabular}{lcc}
\toprule
Method & BIRD & Spider\\
\midrule
Coupled, fixed $\beta$   & 0.142 & 0.109\\
Temperature (post-hoc)   & 0.162 & 0.108\\
Platt (post-hoc)         & 0.021 & 0.026\\
\bottomrule
\end{tabular}
\end{table}

\paragraph*{Edge-shuffle ablation}
Table~\ref{tab:shuffle} asks whether the real foreign-key graph is doing the work, by replacing it with
shuffled edges and refitting nothing. It reports the mean and the 2.5--97.5th percentile range of recall and
weakly-signaled recovery over 60 shuffles, for a random edge permutation and a degree-preserving permutation,
with the field parameters held at their real-graph values.

\begin{table}[h]\centering
\caption{Edge-shuffle ablation for the global fitted model: recall at the gold set size and weakly-signaled
recovery, with the foreign-key edges replaced by shuffled edges. Shuffled rows give the mean and the
2.5--97.5th percentile range over 60 shuffles.}\label{tab:shuffle}
\begin{tabular}{lcc@{\hskip 2em}cc}
\toprule
& \multicolumn{2}{c}{BIRD} & \multicolumn{2}{c}{Spider}\\
\cmidrule(lr){2-3}\cmidrule(lr){4-5}
Edges & Recall & Recovery & Recall & Recovery\\
\midrule
Real foreign keys        & \textbf{0.805} & \textbf{16.5\%} & \textbf{0.916} & \textbf{18.0\%}\\
Random permutation       & 0.777 & 7.1\% $[4,10]$ & 0.896 & 6.9\% $[2,12]$\\
Degree-preserving        & 0.774 & 7.7\% $[5,11]$ & 0.893 & 9.3\% $[3,16]$\\
\bottomrule
\end{tabular}
\end{table}

\paragraph*{Hard-connected variant}
Restricting the exact enumeration to configurations whose selected set is connected gives a hard-connected
variant of the model. At the fitted coupling it improves weakly-signaled recovery markedly (BIRD 16.5\% to
35.5\%, Spider 18\% to 32\%) at slightly higher recall (BIRD 0.805 to 0.828, Spider 0.916 to 0.922) and unchanged
precision. We nonetheless keep the unconstrained model primary for the reason given in Section~2.2, that the hard constraint cannot represent the 5 to 7\% of gold sets that are disconnected.

\paragraph*{A worked example}
Figure~\ref{fig:trace} traces the sequential reading of the model on one BIRD query.

\begin{figure}[h]\centering
\includegraphics[width=0.7\linewidth]{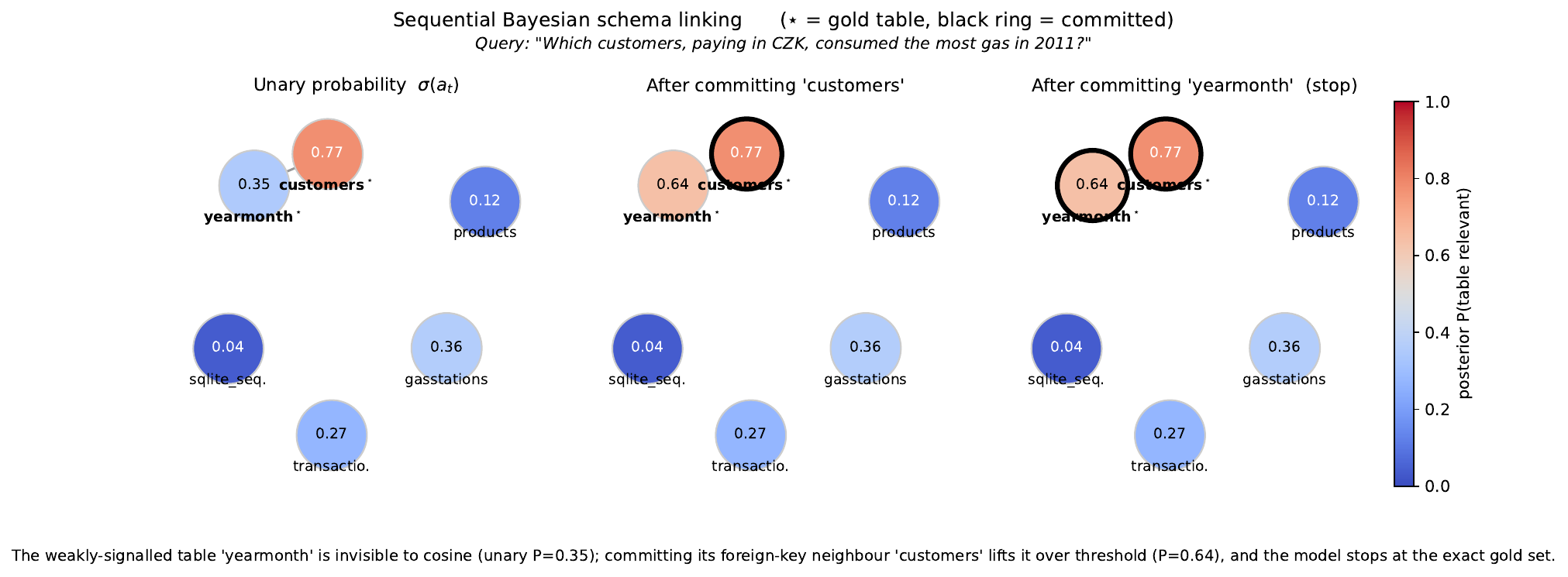}
\caption{Sequential reading of the model on one BIRD query. The weakly-signaled table is scored low by the
cosine (unary probability below one half) and is recovered once its foreign-key neighbor is committed, and the
model stops at the gold set.}\label{fig:trace}
\end{figure}

\paragraph*{Marginal proper scores}
Table~\ref{tab:proper} reports the calibration error, Brier score, and log score behind the marginal-calibration comparison of Section~4.3.

\begin{table}[t]\centering
\caption{Calibration error, Brier score, and log score, cross-fit, with database-clustered 95\% intervals for
the headline comparison. Lower is better throughout.}\label{tab:proper}
\begin{tabular}{lccc@{\hskip 2em}ccc}
\toprule
& \multicolumn{3}{c}{BIRD} & \multicolumn{3}{c}{Spider}\\
\cmidrule(lr){2-4}\cmidrule(lr){5-7}
Method & ECE & Brier & Log & ECE & Brier & Log\\
\midrule
Independent ($\beta{=}0$) & 0.024 & 0.133 & 0.417 & 0.040 & 0.127 & 0.406\\
Coupled, fixed $\beta$    & 0.142 & 0.171 & 0.530 & 0.109 & 0.146 & 0.467\\
Global fit                & 0.022 & 0.134 & 0.421 & 0.023 & 0.129 & 0.412\\
Hierarchical MAP    & \textbf{0.015} & \textbf{0.118} & \textbf{0.376} & 0.029 & \textbf{0.102} & \textbf{0.326}\\
Hierarchical Bayes (predictive) & 0.025 & 0.121 & 0.381 & 0.034 & 0.102 & 0.335\\
\bottomrule
\end{tabular}\\[2pt]
{\footnotesize The first four rows are point-estimate field fits. The last is the full hierarchical
Bayesian posterior predictive of Section 2.3, integrating a Laplace approximation to the parameter
posterior. Clustered 95\% intervals (BIRD): hierarchical ECE $[0.009,0.031]$, coupled ECE $[0.051,0.237]$,
hierarchical log $[0.312,0.475]$, coupled log $[0.403,0.646]$.}
\end{table}

\paragraph*{Outlook}
Large language models are now the default way to turn a question into a query, but the decisions around them,
which tables are relevant, which documents to retrieve, and when to trust an answer, are decisions under uncertainty. This is familiar ground for Bayesian modeling. The model here is one concrete example. It sits
inside a text-to-SQL pipeline, reads the same cheap features the language model already produces, and returns a
calibrated distribution over which schema subgraph to keep. It combines the foreign-key graph, through the
conditional model, with database-level random effects, and reports a distribution rather than a single guess, so
a later step can defer to the user or ask for confirmation when the evidence is thin. It neither replaces nor
retrains the language model but wraps a structured probability model around it.

The construction is not specific to schema linking. Wherever a system must select a small connected substructure
from noisy per-item evidence, and calibrated uncertainty over that structure matters, the same components apply.
Item scores serve as unary evidence, a coupling encodes the known relationships, and a hierarchy lets the coupling
adapt across contexts. Retrieval-augmented generation choosing a connected set of passages in a knowledge graph,
entity linking over a relational store, variable selection on a known dependency graph, and join-path selection in
query optimization all share this shape, with existing model scores entering as covariates in a structured
Bayesian selection model. That such a model can be fit exactly, and that it complements rather than competes with
a black-box language model, is what we take from the schema-linking case.

Two extensions seem most worthwhile. The first is a richer coupling. The pairwise foreign-key term here is a
single attractive scalar, and the same hierarchical fit could carry more structure. Natural options are couplings
that vary with an edge's type or features rather than sharing one value \citep{peterson2015bayesian}, higher-order
or multi-state interactions in the style of a Potts model \citep{moores2020scalable}, and repulsive priors such as
a determinantal point process \citep{kulesza2012determinantal} for settings where the relevant tables should be
spread across the graph rather than clustered. Proposition~3 governs the calibration of the attractive
scalar coupling only, and whether the same intercept correction restores calibration under these
alternatives is open. The second extension is to move beyond relational schemas to other connected-substructure
retrieval, most directly graph-based retrieval-augmented generation, where the target is a connected set of
passages or entities in a knowledge or document graph \citep{subgraphrag,edge2024graphrag}. There the unary
evidence is passage--query relevance and the coupling is the graph's own edges, so little of the construction
would change.

\section{Generality and robustness}\label{s:sec:sim}

\expt{Simulation study}{exp:sim}
A simulation study checks that the mechanism Proposition~3 of the main paper identifies reproduces under
controlled graph density and heterogeneity, rather than being specific to BIRD and Spider. Each
synthetic database is a random connected graph on eight nodes at a target edge density, built from a random
spanning tree plus additional random edges. Each query draws a connected relevant set of size four. A fraction
of its relevant nodes are planted as weakly-signaled, with unary evidence a fixed amount below threshold, while
the remaining relevant nodes and the irrelevant nodes draw Gaussian evidence separated by a signal gap, and
irrelevant nodes adjacent to the relevant set receive a proximity bump so that distractors are structurally
plausible. We fit a calibrated unary model on the evidence, then apply the independent selector and the coupled,
global, and hierarchical field methods exactly as on the real data. Between-database heterogeneity is
introduced by drawing each database's density around the target with a spread parameter. Every condition is run
for 20 independent replications. Table \ref{tab:sim} reports the mean and a 95\% interval across replications.

\begin{table}[t]\centering
\caption{Simulation study, 20 replications per condition. Calibration error for the independent, fixed-coupling,
and hierarchical methods, and weakly-signaled recovery for the fixed-coupling and hierarchical methods.}
\label{tab:sim}
{\small
\begin{tabular}{lcccc@{\hskip 1em}ccc}
\toprule
& \multicolumn{4}{c}{ECE} & \multicolumn{3}{c}{Recovery}\\
\cmidrule(lr){2-5}\cmidrule(lr){6-8}
Condition & Unary & Fixed $\beta$ & Global & Hier. & Fixed $\beta$ & Global & Hier.\\
\midrule
Baseline (density 0.35)      & 0.074 & 0.420 & 0.034 & 0.033 & 50\% & 32\% & 32\%\\
Sparse (density 0.25)        & 0.068 & 0.358 & 0.042 & 0.041 & 52\% & 49\% & 50\%\\
Dense (density 0.70)         & 0.093 & 0.498 & 0.092 & 0.091 & 42\% & \phantom{0}0\% & \phantom{0}1\%\\
Heterogeneous (spread 0.35)  & 0.079 & 0.422 & 0.052 & 0.058 & 50\% & 14\% & 25\%\\
Few queries ($n_q{=}15$)     & 0.080 & 0.420 & 0.055 & 0.058 & 51\% & 17\% & 23\%\\
\bottomrule
\end{tabular}}\\[2pt]
{\footnotesize ``Global'' is the global fitted field (maximum likelihood); ``Hier.'' is the hierarchical fit.
Representative 95\% intervals across replications: hierarchical ECE within about $\pm 0.01$; recovery
$[16,34]$ (hierarchical) against $[7,21]$ (global) at spread 0.35.}
\end{table}

Three patterns from the real data recur, each with its replication uncertainty. First, the calibration pattern of Table
2 reproduces in every condition: the independent selector is calibrated, the coupled posterior at
a fixed coupling is overconfident (ECE 0.36 to 0.50, with intervals narrower than 0.02), and fitting returns the
error to the independent level, so the mechanism is not specific to the two datasets. Second, recovery depends on graph density. In sparse graphs coupling recovers about half of the planted weakly-signaled nodes. In dense graphs
recovery collapses to essentially zero (1\%, interval $[0,2]$), because weakly-signaled nodes become
indistinguishable from spurious connections and the fit responds by reducing the coupling, so that calibration
is preserved (fitted ECE 0.091 against the independent 0.093). The fitted model keeps substantial coupling in
sparse graphs but shrinks the coupling nearly to zero in dense ones. Third, the hierarchical coupling helps most under heterogeneity. When
databases are identical the hierarchical and global fits are indistinguishable (32\% recovery each). When they
differ, the hierarchical fit recovers about twice as many weakly-signaled nodes (25\%, interval $[16,34]$,
against the global 14\%, interval $[7,21]$) at comparable calibration, and the same holds when there are few
queries per database, where partial pooling borrows strength across databases.

\begin{figure}[t]\centering
\includegraphics[width=\linewidth]{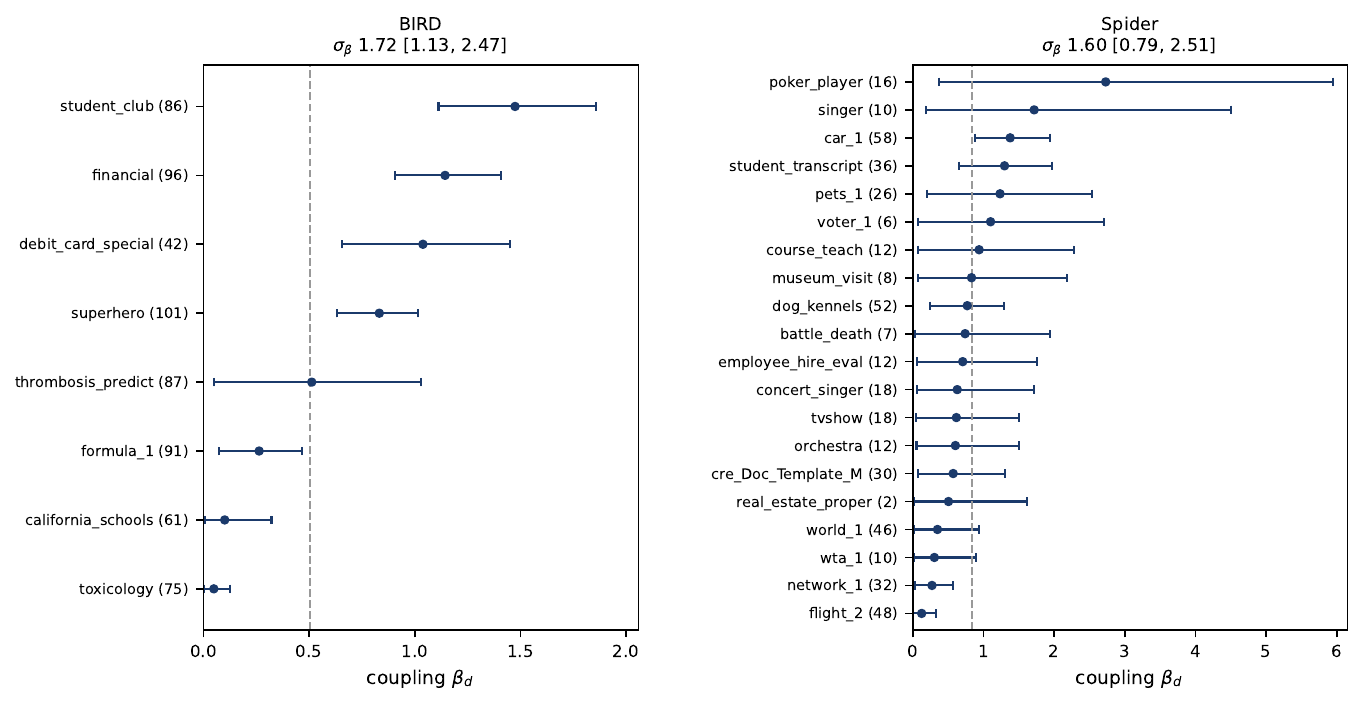}
\caption{The posterior over how much to trust the graph in each database, $\beta_d$, from the
full hierarchical model (Hamiltonian Monte Carlo, a descriptive fit on all queries rather than the cross-fitted
evaluation used elsewhere). Each row is a database, with its query
count in parentheses. The dot is the posterior mean and the bar the 90\% credible interval, and the dashed line
marks the population coupling. The couplings differ widely across databases, and the intervals widen for
databases seen in fewer queries.}\label{fig:coupling}
\end{figure}

\expt{Uncertainty over how much to trust each graph}{s:exp:coupling}
The simulation shows that a hierarchical coupling helps when databases differ. To see whether these databases
differ, and with what certainty, we fit the full model to all queries by Hamiltonian Monte Carlo and examine the
posterior over each database's coupling $\beta_d$, that is, over how much to trust its foreign-key graph.
Figure~\ref{fig:coupling} shows these posteriors. The couplings range widely, from databases with strong graph
reliance (BIRD's \texttt{student\_club}, $\beta_d\approx1.4$) to databases where the coupling is near zero and
tables are scored on their own evidence (\texttt{toxicology}, $\beta_d\approx0.05$).
The posterior for the population spread $\sigma_\beta$ concentrates near 1.6 to 1.7, and its 90\% interval has a
lower end of about 1.1 on BIRD and 0.8 on Spider, concentrated away from zero and consistent with substantial
between-database heterogeneity rather than merely extra parameters. The intervals are also wider for databases
with fewer queries. On Spider the width of a database's interval falls with its number of queries (correlation $-0.46$), so a database
seen in only a handful of queries, such as \texttt{poker\_player} with sixteen, gets a wide interval and is
pulled toward the population coupling, while a data-rich database such as \texttt{car\_1} gets a tight one. These
posterior intervals quantify uncertainty in the database-specific coupling, which a plug-in estimate cannot
provide.

\expt{Unseen databases}{s:exp:lodo}
Table \ref{tab:lodo} evaluates a genuinely unseen database with a full leave-one-database-out protocol, in which
both the unary weights and the field are trained on the other databases and the held-out database is predicted
from the population effect. The structural benefit transfers, since the model recovers about 15\% of
weakly-signaled tables against zero for the independent selector while recall improves. Calibration transfers only partially, and
substantially better on Spider than on BIRD (0.084 against the independent 0.080 on Spider, versus 0.113 against
0.074 on BIRD). Spider supplies more training databases (19 against 7), one plausible explanation consistent with
the partial-pooling mechanism, though the two benchmarks differ in other respects as well. Averaging over the coupling's
uncertainty (the full predictive) makes the marginals more overconfident, not less, and scores worse on
both proper scores (Table \ref{tab:lodo}). In these fits, integrating over the right-skewed distribution of the
positive coupling places additional weight on relatively large values of $\beta_d$, raising the marginal
inclusion probabilities rather than hedging them. On an unseen database the plain
population plug-in is the better predictor on every metric we report, a more guarded conclusion than the
database-seen setting supports.

\begin{table}[t]\centering
\small
\setlength{\tabcolsep}{4pt}
\caption{Leave-one-database-out (database unseen): unary weights and field both trained on the other databases,
held-out database predicted from the population effect. Calibration error, Brier score, log score, and recall;
lower is better for the first three. Weakly-signaled recovery (omitted) is 13--17\% for the unseen conditions
against 0\% for the independent selector.}\label{tab:lodo}
\begin{tabular}{lcccc@{\hskip 1.5em}cccc}
\toprule
& \multicolumn{4}{c}{BIRD} & \multicolumn{4}{c}{Spider}\\
\cmidrule(lr){2-5}\cmidrule(lr){6-9}
Condition & ECE & Brier & Log & Recall & ECE & Brier & Log & Recall\\
\midrule
Independent ($\beta{=}0$)   & 0.074 & 0.148 & 0.461 & 0.773 & 0.080 & 0.151 & 0.468 & 0.895\\
Unseen, population plug-in  & 0.113 & 0.157 & 0.486 & 0.796 & 0.084 & 0.151 & 0.483 & 0.915\\
Unseen, full predictive     & 0.178 & 0.182 & 0.547 & 0.799 & 0.113 & 0.178 & 0.535 & 0.908\\
\bottomrule
\end{tabular}
\end{table}

\expt{Three warehouses, and what survives across them}{exp:warehouses}
The warehouse study above used one schema. The same release carries five, and two of them are better matched
to the mechanism than the one we had been using: \textsf{neutron} has $175$ tables, $160$ declared foreign-key
edges, $30$ multi-table questions with a mean gold set of $7.5$ tables, and $29$ of $30$ gold sets lying in a
single foreign-key component. \textsf{nova} has $109$ tables but only $24$ declared edges and a largest
component of $17$. Notably \textsf{dw}, the schema of the main text, declares \emph{no} foreign keys
in its catalogue at all, which is why that study needs a separate join-key file. We report all three, with the
same arms and the same metric.

One thing has to be said first, because it limits what can be measured. On \textsf{neutron} the semantic field
is weak enough that whole-gold-set containment is out of reach at any enumerable pool size: gold tables have
median rank $21$ of $175$, but the \emph{worst} gold table in a question has median rank $159$, and no question
has its whole gold set inside even the top $60$. Containment, the deployable metric of
the main text, is therefore unavailable on a real $175$-table warehouse, and we fall back to
recall at the true gold size as in the main text.

\begin{table}[t]\centering\small
\caption{Recall at the true gold-set size on the two additional warehouses, with the arms of
the main text. The best \emph{unary} baseline changes between databases. The coupled posterior does
not have to be chosen.}\label{s:tab:warehouses}
\begin{tabular}{lccc}
\toprule
arm & \textsc{neutron} & \textsc{nova} & worst case\\
\midrule
cosine            & $\mathbf{0.346}$ & $0.212$ & $0.212$\\
unary fusion      & $0.167$ & $0.446$ & $0.167$\\
foreign-key closure & $0.201$ & $0.446$ & $0.201$\\
PageRank diffusion  & $0.281$ & $\mathbf{0.744}$ & $0.281$\\
coupled, top-$K$ pool & $0.277$ & $0.590$ & $0.277$\\
coupled, bounded closure & $0.341$ & $0.724$ & $\mathbf{0.341}$\\
\bottomrule
\end{tabular}
\end{table}

The table contains one result we emphasize and one we do not. The one we would not:
the coupled posterior does not reliably separate from PageRank diffusion. It is ahead by $+0.059$
$[+0.005,+0.122]$ on \textsf{neutron}, behind by $-0.020$ $[-0.079,+0.048]$ on \textsf{nova}, and ahead by
$+0.013$ $[-0.013,+0.040]$ on \textsf{dw}. Across three warehouses a simple diffusion baseline is a match for
the structured posterior on ranking accuracy, and we say so.

The one we would emphasize is about robustness rather than peak accuracy. Which \emph{unary} baseline is best
flips completely between the two databases: cosine is far ahead on \textsf{neutron} ($0.346$ against $0.167$)
and far behind on \textsf{nova} ($0.212$ against $0.446$). Choosing a unary arm therefore requires knowing
which database one is on, and choosing wrong costs about half the achievable recall in either direction. The
coupled posterior is within $1.5\%$ of the best arm on \textsf{neutron} and $62\%$ above the best unary arm on
\textsf{nova}, giving it the best worst case of any arm in the table, $0.341$ against $0.281$ for the next
best. That is the same claim the paper makes elsewhere in a different metric: estimating how far to trust the
graph is worth less than the best oracle choice on any single dataset and more than any fixed choice across
datasets. Two qualifications keep it the right size. A worst case over three schemas is a descriptive
minimax and carries no useful interval, so it should be read as a stress test of robustness rather than a
deployment result. And recall at the true gold size is an oracle ranking metric, used here only because
containment is unavailable at this scale. It is not the deployable quantity of
the main text.

\expt{When the join graph is not given}{s:exp:uncertain}
Every result so far conditions on a known $G$, read from declared foreign keys. That is a curated
artifact, and a real warehouse frequently does not ship it: of the $28$ Spider~2.0 databases,
\emph{none} declares a single foreign key. The assumption is therefore a limitation of the method and
not only of one benchmark, and it is worth attacking rather than noting.

The natural repair is to stop conditioning on $G$ and put a model on it. For each unordered table
pair we compute evidence available in any DDL, shared column names, the $\langle$table$\rangle$\_id
naming convention, column-token and table-name overlap, and value-token overlap where the data are
reachable. A logistic edge model is fitted on declared foreign keys leave-one-database-out, and the
energy uses its prediction as a weight,
\[
   \log p(S\mid a,w)=\sum_{t\in S}a_t+\beta\sum_{t<u}w_{tu}x_tx_u+c(|S|),
\]
which is the same model with the $0/1$ adjacency replaced by $w\in[0,1]$. Exact enumeration is
untouched: it already sums over all $2^n$ masks and a weighted pair term costs no more than a binary
one. Because BIRD and Spider \emph{do} declare foreign keys, we can hide them, infer them, and price
what the inference costs against the declared-graph ceiling.

\begin{table}[t]\centering\small
\caption{Containment at a fixed budget with the join graph inferred rather than declared, BIRD, true
leave-one-database-out throughout, including the edge model. Gains are over the independent
selector.}\label{s:tab:uncertain}
\begin{tabular}{lcccl}
\toprule
graph & $k{=}2$ & $k{=}3$ & $k{=}4$ & gain at $k{=}3$\\
\midrule
independent (no graph)        & $0.432$ & $0.712$ & $0.814$ & ---\\
declared foreign keys         & $0.432$ & $0.746$ & $0.858$ & $+0.0344$ $[+0.0110,+0.0579]$\\
inferred, weighted            & $0.423$ & $0.732$ & $0.850$ & $+0.0203$ $[-0.0031,+0.0454]$\\
inferred, thresholded         & $0.426$ & $0.739$ & $0.856$ & $+0.0266$ $[+0.0047,+0.0485]$\\
inferred, declared removed    & $0.432$ & $0.712$ & $0.814$ & $+0.0000$\\
\bottomrule
\end{tabular}
\end{table}

The edge model reaches leave-one-database-out AUC $0.708$ on average and $0.876$ at the median, the
mean pulled down by one database whose column naming carries no convention. Table~\ref{s:tab:uncertain}
gives the consequence. The inferred graph recovers $0.59$ of the declared-edge gain and, thresholded,
is itself significant, $+0.0266$ $[+0.0047,+0.0485]$, and inferred minus declared is $-0.0141$
$[-0.0329,+0.0047]$, so on this evidence we cannot distinguish the two. The practical reading is that
the method does not require a curated join graph, which removes an assumption rather than adding a
result.

Two parts of the same table are negative and we report them as such. Weighting does not beat
thresholding, $-0.0063$, so the posterior edge probability is doing nothing that a hard cut at one
half does not already do, and we do not claim the weighted form. And with the declared edges deleted
from the candidate set the gain is exactly zero, the coupling selecting $\beta=0$ of its own accord:
the edge model rediscovers the DDL rather than finding join structure the DDL omits. The contribution
is that the graph can be reconstructed, not that a better one can be found.

Spider cannot arbitrate this. Its declared-edge gain at $k{=}3$ is itself $-0.0131$, so there is no
ceiling to recover, and the comparison is uninformative there rather than failed, for the reason
the main text gives: a four-table median schema leaves a budget of three retaining nearly
everything.

\paragraph*{Implementation}
Embeddings are $1536$-dimensional, from the standard pretrained OpenAI model
\texttt{text-embedding-3-small}. Each table is described by the string ``name: column$_1$, column$_2$, \dots'', and each question by its raw
text. The six features in $\phi_t$ are the query--table cosine, the maximum query--column cosine over the
table's columns, a BM25 score, a token-overlap fraction against the table's name, a token-overlap fraction
against its columns, and a value-match fraction that fires when a question token appears among the table's stored
values, from a bounded scan of up to $500$ rows of the training databases. Features are standardized before the logistic fit. In the full Bayesian model
$\theta$ carries no separate intercept, since $\alpha_d$ is the model intercept. The fixed logistic unary used by
the independent and coupled baselines includes its own intercept inside $a_t$, so those baselines keep a base-rate
term even when the field intercept $\alpha$ is set to zero. The priors are
$\theta_j\sim\N(0,4)$, $\alpha_0,b_0\sim\N(0,9)$, and $\sigma_\alpha,\sigma_\beta\sim\text{Half-Normal}(1)$, with
$\operatorname{softplus}(x)=\log(1+e^x)$. The Laplace approximation is formed in unconstrained coordinates, with
the positive scales $\sigma_\alpha,\sigma_\beta$ on the log scale, and the posterior predictive uses $200$ draws, a count validated in \hyperref[s:app:checks]{Supplement~\ref*{s:app:checks}}.

\section{Ablations and negative results}\label{s:sec:supp-neg}
\expt{Edge-shuffle ablation}{exp:shuffle}
A weakly-signaled table is defined by the unary model missing it, so we test whether its recovery uses the real
foreign-key structure by re-running the model with the graph edges shuffled and the field parameters held at
their real-graph values. Replacing the graph by a random or degree-preserving shuffle collapses weakly-signaled
recovery from 16--18\% to roughly 7--9\%, and the real-edge value lies above the entire 60-shuffle range, a
one-sided randomisation comparison, so the gain is structural rather than an effect of refitting
(Appendix~\ref{s:app:extra}, the corresponding table of the main text). Weakly-signaled tables are common. They occur as 279 gold-table instances across the 639
BIRD queries and 50 across the 459 Spider queries.

\expt{Connector recovery}{exp:connector}
The edge-shuffle result speaks to weakly-signaled tables in general. We can also single out the connector
tables, the articulation vertices of the gold-induced subgraph (the junctions of Figure 1 of the main text). There
are 168 of them across the BIRD queries and 68 across Spider, and they are
mostly a different group from the weakly-signaled tables. A connector is often easy to score on its own, and
only 28\% of BIRD connectors and 1\% of Spider connectors are also weakly-signaled. The connector column of
the selector comparison of the main text reports how often each method places these tables in its top-$|\R|$ ranking, the same
rule used for recall. On BIRD, where many connectors are only moderately salient, coupling lifts connector recovery from 71\%
for the independent ranking to 87\% for the hierarchical fit, with database-clustered 95\% bootstrap intervals of
$[59,87]$ and $[72,95]$. On Spider the independent ranking already recovers 67 of the 68 connectors, so little is
left to gain and the fitted models reach 100\%. The gain is therefore larger on BIRD, where more connectors are
only moderately salient.

\expt{Recovery--precision tradeoff}{exp:phase}
Sweeping the coupling varies $\beta$ over a grid. Recall and recovery rise while precision falls, and the fitted
coupling sits where recovery is high before precision collapses. The aggregate curve is consistent with the
tradeoff Proposition 1 of the main text predicts, recovery rising once coupling is strong enough to rescue weak
tables and precision declining once it grows too large, though it aggregates over queries rather than checking
each query's own window.

\expt{The junction-table explanation is wrong}{exp:role}
the weakly-signaled recovery analysis of the main text locates the coupling's gain on queries containing a weakly-signaled gold
table but does not say why those tables are weakly signaled. A normalized schema offers an obvious
answer. A many-to-many relation is realized by a \emph{junction} table whose columns are almost all
foreign keys, \textsf{student\_course(student\_id, course\_id)}, which by construction carries almost
no descriptive text for a question to match however good the embedding is, while sitting between two
tables that do match. That is precisely the bridge configuration the pair term exploits. If the story
holds, junction tables should be over-represented among weakly-signaled gold tables, and telling the
model which tables are junctions should pay.

We pre-registered both halves. The role covariate is computed from the schema alone, $r_t=1$ when a
table has foreign-key degree at least two and at least two thirds of its column names end in
\textsf{id}, \textsf{key} or \textsf{code}. It enters the field as $\sum_t(a_t+\gamma r_t)x_t$ with
$\gamma$ fitted leave-one-database-out beside $\beta$, and a path-aware variant adds a term for
two-paths whose midpoint is a junction.

Both halves fail, and the descriptive half fails in the informative direction. Junction tables are
$20.1\%$ of BIRD gold tables but only $10.5\%$ of the weakly-signaled ones, an enrichment of
$0.52$. On Spider they are $14.1\%$ of gold tables and $0$ of $46$ weakly-signaled ones. Junction
tables are thus \emph{under}-represented exactly where the mechanism predicted a surplus, and on
Spider the prediction is refuted outright. Consistently, the role term buys nothing: at $k{=}3$ the
gain over plain coupling is $+0.0000$ on both datasets, and on the weakly-signaled stratum too,
because the leave-one-database-out selection sets $\gamma=0$ and the two-path weight to zero in every
fold. The model declines the covariate rather than being hurt by it.

We report this as a second null of the same kind as the degree term reported earlier,
and the pair of them supports one reading: the schema-side role information available to us is
already carried by the edges, and adding it to the field is redundant. It also leaves a question
open that we would rather leave open than answer badly. We do not know why the rescued tables are
weakly signaled. We now know it is not because they are junctions.

\expt{The obvious deployable gate does not survive an honest protocol}{exp:gate}
The heterogeneity above is an oracle finding: it needs the gold set. Its practical content depends on whether
the same condition can be detected from observables, so we tried. Coupling helps when a table with weak unary
evidence sits next to tables with strong unary evidence, and every term there is available at prediction time,
which gives the gate: fire when some table ranked below the budget with $\sigma(a_t)<\tau$ has a foreign-key
neighbor ranked inside it, use the coupled posterior then and the independent model otherwise, with $\tau$
fitted on other databases.

Under a protocol that fits the coupling leave-one-database-out but the unary field by query fold, the gate
looks like the best policy in this paper: containment $0.7778$ against $0.7653$ for coupling everywhere and
$0.7293$ for the independent model. Under a protocol that fits the \emph{unary field} leave-one-database-out
as well, it is not: $0.7418$ against $0.7480$ and $0.7121$, a difference of $-0.0063$
$[-0.0172,+0.0031]$. The gate does not beat coupling everywhere, and we report it as a negative.

The asymmetry is instructive and applies beyond this gate. The stratified comparisons of Analysis~6 of the main text are differences between two arms that share one unary field, so
they are almost unchanged by how that field is fitted: the schema-size gain at seven or more tables is
$+0.0775$ $[+0.0374,+0.1176]$ under either protocol, and the weakly-signaled gain moves only from $+0.1365$
to $+0.1349$. A gate is different in kind. It \emph{reads} the unary field to make a decision, so it inherits
whatever optimism that field carries, and a field fitted with same-database queries is more trustworthy at
selection time than one fitted without them. Any query-adaptive rule built on unary scores has to be validated
with the scores fitted as they would be at deployment. Ours was not, and the effect it appeared to have was
that difference.

What survives is the negative: the oracle stratum is real and large, and this route does not convert it into an
operational gain. Whether some other observable does remains open.

\end{document}